\documentclass[a4paper, 11pt]{amsart}
\usepackage{amsmath, amssymb, fontenc, amsthm, array, hyperref, amsfonts}
\usepackage[english]{babel}
\usepackage{cite}
\usepackage[T1]{fontenc}
\usepackage{tgtermes}
\usepackage{mathptmx}  
\usepackage[scaled=.92]{helvet}
\usepackage{mathrsfs}
\usepackage[numbers]{natbib}
\usepackage{hyperref}
\usepackage{mathtools}
\usepackage{stmaryrd}
\usepackage{enumitem}
\usepackage{microtype}

\hypersetup{
    colorlinks = true,
    citecolor = blue,
    urlcolor = blue,
    linkcolor = red
}

\renewcommand{\keywords}[1]{\footnotesize \textbf{\textit{Keywords:}} #1 \normalsize}
\renewcommand{\restriction}{\mathord{\! \upharpoonright}}

\newtheorem{theorem}{Theorem}[section]
\newtheorem{lemma}[theorem]{Lemma}

\newtheorem{corollary}[theorem]{Corollary} 
\newtheorem{proposition}[theorem]{Proposition} 

\theoremstyle{definition}
\newtheorem{definition}[theorem]{Definition}

\theoremstyle{remark}

\usepackage{tikz}
\usepackage{graphicx}
\usetikzlibrary{matrix,arrows.meta, shapes.misc,fit,decorations.pathreplacing}

\tikzset{arrow/.style = { thick, color=black, ->, >=Triangle}, }

\title{Importing soundness and completeness in modal logics}

\author{Pedro Teixeira Yago}\thanks{Orcid: 0000-0001-7993-4516, Classe di Lettere e Filosofia, Scuola Normale Superiore di Pisa, Italy. \texttt{pedro.tyago@outlook.com}}

\begin{document}

\maketitle

\begin{abstract}
We develop general strategies for transferring soundness and completeness from more expressive modal languages to less expressive ones, unifying several existing notions of operator definability along the way. For soundness, we exploit semantic insensitivity: if a less expressive language is insensitive to a frame operation, soundness extends to the operation's closure of the original frame class. For completeness, restricting to relational semantics and languages with a single operator, we present strategies for relating the target logic's canonical model to that of a normal modal logic via a truth-preserving translation. Some of those dispense entirely with specifying an accessibility condition for the target logic, inheriting it from a normal modal logic instead.
\end{abstract}


\section{Introduction}\label{introduction}

A recurring challenge in modal logic concerns logics whose primitive operators are not the standard modal operator $\Box$, with the standard semantics. The non-contingency operator $\Delta$ \cite{humberstone1995}, the essence operator $\circ$ \cite{marcos2005}, the being-wrong operator $W$ \cite{steinsvold2011}, and the unknowable-necessary-truth operator $U$ \cite{yagoventuri2026} are all natural objects of study, yet none of them is the standard necessity operator. The logics built around such operators are, in a precise sense to be made explicit shortly, less expressive than the minimal normal modal logic $\mathbf{K}$.

In the normal modal setting, soundness and completeness are routinely established by the correspondence between axioms and frame properties, and completeness is obtained through canonical models. For logics in less expressive languages, however, canonical model constructions require a non-obvious specification of the accessibility condition -- one whose definition, as illustrated across the literature, is case-by-case and sometimes technically involved.

The aim of this paper is to develop strategies for transferring soundness and completeness results to such logics from more expressive ones, exploiting the structural relationship between the languages. The natural and most important instance of this relationship is the one between a non-standard modal language and the normal modal language: the aforementioned operators are definable in terms of $\Box$ in appropriate frame classes, which places their languages below the normal modal language in terms of expressivity.

First, we develop soundness in Section \ref{sec soundness}. One of the key concepts there is semantic insensitivity: a modal language is insensitive to a frame property -- or to a frame operation -- when it cannot distinguish between frames that differ with respect to (henceforth, \emph{w.r.t.}) $P$, or is invariant to $\rho$. If a less expressive language $\mathcal{L}_2$ is insensitive to an operation $\rho$, and if $\mathbf{L}_h$ denotes the image of a logic $\mathbf{L}$ in a more expressive language $\mathcal{L}_1$ under a truth-preserving translation $h : \mathcal{L}_1 \to \mathcal{L}_2$, then $\mathbf{L}_h$ is sound not only w.r.t. the class of frames $\mathbf{L}$, but also w.r.t. the closure of the class under $\rho$. The work is done in the more expressive language. We stress that the framework of Section \ref{sec soundness} is general, in that it applies to any pair of modal languages one of which is less expressive than the other, thus making no commitment to one of them being the normal modal language. The results subsume and generalise existing methods from \cite{goldblattmares2006} and \cite{gilbertventuri2016}, where specific instances of this strategy were developed for the logics of essence and contingency.

Section \ref{sec completeness} concerns completeness. Unlike the strategies for soundness, this section is specific to relational semantics and logics with a single primitive unary operator. We present three strategies for completeness -- and one more which is a brief modification of another. The first generalizes the work in \cite{goldblattmares2006} and \cite{gilbertventuri2016}, which uses bounded morphisms between the canonical model of a normal modal logic and the canonical model of its translation to a different modal language, under the assumption of some constraints on the translation. The two -- or three -- other strategies sidestep the construction of a canonical model for logics of the non-normal modal language altogether. Standard completeness proofs for non-normal logics must specify, for each world of the canonical model, an accessibility condition that allows one to prove the usual truth lemma. We circumvent this specification by building directly on the canonical model of $\mathbf{K}$. Given a truth-preserving and theorem reflecting translation from the non-normal modal language to the normal one, the maximal consistent sets of the image logic form a natural submodel of the canonical frame of $\mathbf{K}$. Completeness of the less expressive logic then follows either from the canonicity of the translation of the normal modal logic by standard frame morphism results, or from the completeness of the non-normal modal logic.

We note an asymmetry between the two strategies that we regard as an important open problem. The soundness transfer is grounded in insensitivity: it is precisely the non-normal language's inability to distinguish between frame classes that makes the transfer possible, and this inability is what the results exploit. The completeness transfer, by contrast, proceeds by canonical model construction and does not, at present, admit an analogous insensitivity-based explanation. The connection between insensitivity and completeness remains open.

The paper is structured as follows. Section \ref{section insens} reviews and unifies four existing notions of definability for modal operators, proving them equivalent. This unification allows one to speak interchangeably about the expressivity of a language and the definability of its primitive operator in another language, and introduces the translation apparatus on which both strategies are based. Section \ref{sec soundness} introduces semantic insensitivity and develops the general soundness import strategy, with results formulated for both semantic and proof-theoretic notions of logic. Section \ref{sec completeness} develops the completeness strategy in the setting of relational semantics and single-operator logics. Section \ref{sec application} offers an application of the strategies developed to the concrete case of a new modal operator. More interestingly, we show the completeness of the minimal logic of the language of the operator $\circ$ solely by means of defining an adequate translation from the language $\mathcal{L}^\circ$ to the normal modal language $\mathcal{L}^\Box$ -- therefore circumventing completely the need for defining a canonical model for the logics of the non-normal language. We end, in Section \ref{sec conclusion}, with some concluding remarks.

\section{Expressivity, definability and insensitivity}\label{section insens}

The goal of this paper is to motivate the study of insensitivity for different modal languages by providing a proof strategy that allows to import soundness and completeness results from a logic in a language $\mathcal{L}_1$ to its counterpart in a target language $\mathcal{L}_2$. The context where we will work is that of modal languages with a single, primitive, \emph{unary}  operator. Moreover, we will work in the standard relational semantics. Central to our investigation is the notion of translation between modal languages, so we first review and unify different approaches that we find in the literature. 

\begin{definition}[Expressivity \cite{ditmarsch2007}]\label{definition expressivity}
We say $\mathcal{L}_2$ is \emph{as expressive as $\mathcal{L}_1$ in $\mathbb{C}$} \linebreak ($\mathcal{L}_1 \preceq \mathcal{L}_2$), if for all $\varphi \in Form_{\mathcal{L}_1}$ there is $\psi \in Form_{\mathcal{L}_2}$ such that for all $\mathcal{M} = \langle W, R, V \rangle$ such that $\langle W, R \rangle \in \mathbb{C}$ and $w \in W$,

\begin{center}
$\mathcal{M}, w \models \varphi$ iff $\mathcal{M}, w \models \psi$.
\end{center}

\noindent
If $\mathcal{L}_1 \preceq \mathcal{L}_2$ but $\mathcal{L}_2 \not\preceq \mathcal{L}_1$ in $\mathbb{C}$, then we say $\mathcal{L}_1$ is \emph{less expressive than $\mathcal{L}_2$ in $\mathbb{C}$} ($\mathcal{L}_1 \prec \mathcal{L}_2$). 
\end{definition}

If $\mathcal{L}_1 \preceq \mathcal{L}_2$ in the class of all frames, we may say $\mathcal{L}_2$ is less expressive than $\mathcal{L}_1$, simpliciter. Notice it may always happen that they are equally expressive in a smaller class of frames.\footnote{For example, $\mathcal{L}^\Box$ and the language of the logics of essence and accident, $\mathcal{L}^\circ$ \cite{marcos2005}, are equally expressive in $\mathbb{C}_\mathbf{T}$.}

We also define the notion of definability of modal operators. Many different definitions have been given, which we here try to generalize. In \cite{marcos2005}, we have the following.

\begin{definition}[Definability by schema]\label{def definability 1}
We say $\Box_1$ is \emph{definable by schema in $\mathcal{L}_2$ w.r.t. $\mathbb{C}$} if there is a truth-preserving translation $h: \mathcal{L}_2 \rightarrow \mathcal{L}_1$ on $\mathbb{C}$ and \linebreak $\lambda: \mathcal{L}_2 \rightarrow \mathcal{L}_2$ such that, for any $\varphi \in Form_{\mathcal{L}_1} \cap Form_{\mathcal{L}_2}$,

\begin{center}
$\mathbb{C} \models \Box_1 \varphi \leftrightarrow h(\lambda(\varphi))$.
\end{center}
\end{definition}

\noindent
This definition reflects the fact the modal operator of a language may be defined by a scheme in another language.

In \cite{gilbertventuri2016}, a slightly different definition is given, which we generalise here.

\begin{definition}[Definability by translation]\label{def definability 2}
We say $\Box_1$ is \emph{definable by translation in $\mathcal{L}_2$ w.r.t. $\mathbb{C}$} if there is a truth-preserving translation $h: \mathcal{L}_2 \rightarrow \mathcal{L}_1$ on $\mathbb{C}$ and $\lambda: \mathcal{L}_2 \rightarrow \mathcal{L}_2$ such that, for any $\varphi \in Form_{\mathcal{L}_2}$,

\begin{center}
$\mathbb{C} \models \Box_1 h(\varphi) \leftrightarrow h(\lambda(\varphi))$.
\end{center}
\end{definition}

An alternative, and more commonly used, definition of definability is this:

\begin{definition}[Definability on the kernel]\label{def definability 3}
We say $\Box_1$ is \emph{definable on the kernel in $\mathcal{L}_2$ w.r.t. $\mathbb{C}$} if, there is $\lambda: \mathcal{L}_2 \rightarrow \mathcal{L}_2$ such that, for any $\varphi \in Form_{\mathcal{L}_1} \cap Form_{\mathcal{L}_2}$, $\mathcal{M} = \langle W, R, V \rangle$ with $\langle W, R \rangle \in \mathbb{C}$ and $w \in W$,

\begin{center}
$\mathcal{M}, w \models \Box_1 \varphi \leftrightarrow \mathcal{M}, w \models \lambda(\varphi)$.
\end{center}
\end{definition}

\noindent
This sort of definition is supposed to reflect what we mean when we say informally, for example, $\circ \varphi$ is definable as $\varphi \rightarrow \Box \varphi$. A slightly stronger alternative definition of definability is the following.

\begin{definition}[Definability by extension]\label{def definability 4}
Let $\mathcal{L}_2^{\Box_1}$ be $\mathcal{L}_2$ extended with $\Box_1$. We say $\Box_1$ is \emph{definable by extension in $\mathcal{L}_2$ w.r.t. $\mathbb{C}$} if there is $\lambda: \mathcal{L}_2^{\Box_1} \rightarrow \mathcal{L}_2$ such that, for any $\varphi \in Form_{\mathcal{L}_2^{\Box_1}}$,

\begin{center}
$\mathbb{C} \models \Box_1 \varphi \leftrightarrow \lambda(\varphi)$.
\end{center}
\end{definition}

\noindent
This definition reflects the possibility of defining the modal operator of a language as a non-primitive modal operator of another.

As we should expect, those notions are related in some way.

\begin{theorem}\label{theorem definabilities 1}
$\Box_1$ is definable on the kernel in $\mathcal{L}_2$ w.r.t. $\mathbb{C}$ iff it is definable by extension in $\mathcal{L}_2$ w.r.t. $\mathbb{C}$.
\end{theorem}

\begin{proof}
Let $\Box_1$ be definable on the kernel in the manner described, and $\lambda : \mathcal{L}_2 \to \mathcal{L}_2$ be the relevant mapping. Define $v: \mathcal{L}_2 ^{\Box_1} \to \mathcal{L}_2$ as:

\begin{tabular}{lll}
$v(p)$ & $=$ & $p$\\
$v(\neg \varphi)$ & $=$ & $\neg v(\varphi)$\\
$v(\varphi \wedge \psi)$ & $=$ & $v(\varphi) \wedge v(\psi)$\\
$v(\Box_1 \varphi)$ & $=$ & $\lambda(v(\varphi))$\\
$v(\Box_2 \varphi)$ & $=$ & $\Box_2 v(\varphi)$
\end{tabular}

\noindent
We show $v$ is truth-preserving on $\mathbb{C}$ by an induction on the complexity of \linebreak $\varphi \in Form_{\mathcal{L}_2^{\Box_1}}$. The atomic case is straightforward. By induction hypothesis, the non-modal cases are easily obtained. If $\mathbb{C} \models \Box_2 \varphi$, then by induction hypothesis and uniform substitution, that is the case iff $\mathbb{C} \models \Box_2 v(\varphi)$. If $\mathbb{C} \models \Box_1 \varphi$, by Definition \ref{def definability 3}, that is the case iff $\mathbb{C} \models \lambda (\varphi)$, so by induction hypothesis and uniform substitution, that is the case iff $\mathbb{C} \models \lambda (v(\varphi))$. Therefore, $\lambda \circ v$ is the desired mapping, and $\Box_1$ is definable by extension in $\mathcal{L}_2$ w.r.t. $\mathbb{C}$. Since $Form_{\mathcal{L}_1} \cap Form_{\mathcal{L}_2} \subseteq Form_{\mathcal{L}_2}$ and $\mathcal{L}_2 \subseteq \mathcal{L}_2^{\Box_1}$, the other direction is straightforward.
\end{proof}

\begin{theorem}\label{theorem definabilities 2}
$\Box_1$ is definable by schema in $\mathcal{L}_2$ w.r.t. $\mathbb{C}$ iff it is definable on the kernel in $\mathcal{L}_2$ w.r.t. $\mathbb{C}$.
\end{theorem}

\begin{proof}
Let $\Box_1$ be definable on the kernel as described. Define $t: \mathcal{L}_1 \to \mathcal{L}_2$ as:

\begin{tabular}{lll}
$t(p)$ & $=$ & $p$\\
$t(\neg \varphi)$ & $=$ & $\neg t (\varphi)$\\
$t(\varphi \wedge \psi)$ & $=$ & $t(\varphi) \wedge t(\psi)$\\
$t(\Box_1 \varphi)$ & $=$ & $\lambda (t(\varphi))$
\end{tabular}

\noindent
By a quick induction on the complexity of formulas, and Definition \ref{def definability 3}, we see $t$ is truth-preserving. Let $t^{-1}$ be its inverse. Clearly, $t$ may not be surjective, so that $t^{-1}$ is not defined for all formulas of $\mathcal{L}_2$. It is, however, defined for any $\varphi \in Form_{\mathcal{L}_1} \cap Form_{\mathcal{L}_2}$. Since $t$ is truth-preserving, so is $t^{-1}$, and hence we have \linebreak $\mathbb{C} \models \Box_1 \varphi$ iff $\mathbb{C} \models t^{-1}(\lambda (\varphi))$, for any $\varphi$ in that intersection. The other direction is straightforward from Definitions \ref{def definability 1} and \ref{def definability 3}.
\end{proof}

\begin{theorem}\label{theorem definabilities 3}
$\Box_1$ is definable by translation in $\mathcal{L}_2$ w.r.t. $\mathbb{C}$ iff it is definable by schema in $\mathcal{L}_2$ w.r.t. $\mathbb{C}$.
\end{theorem}

\begin{proof}
Suppose $\Box_1$ is definable by translation, and let $h$ and $\lambda$ be the relevant translation and function. Let $\varphi \in Form_{\mathcal{L}_1} \cap Form_{\mathcal{L}_2}$. Since $h$ is truth-preserving, $\mathbb{C} \models h(\varphi) \leftrightarrow \varphi$, so $\mathbb{C} \models \Box_1 h(\varphi) \leftrightarrow \Box_1 \varphi$. By assumption, $\mathbb{C} \models \Box_1 h(\varphi) \leftrightarrow h(\lambda(\varphi))$, and therefore $\mathbb{C} \models \Box_1 \varphi \leftrightarrow h(\lambda(\varphi))$. Let now $\Box_1$ be definable by schema, $\lambda$ be the relevant function, and $h$ the relevant translation. Let $h': \mathcal{L}_2 \to \mathcal{L}_1$ be defined as:

\begin{tabular}{lll}
$h'(p)$ & $=$ & $h(p)$\\
$h'(\neg \varphi)$ & $=$ & $\neg h'(\varphi)$\\
$h'(\varphi \wedge \psi)$ & $=$ & $h'(\varphi) \wedge h'(\psi)$\\
$h'(\lambda (\varphi))$ & $=$ & $\Box_1 h'(\varphi)$
\end{tabular}

\noindent
Since $h$ is truth-preserving, by induction on the complexity of formulas we may see $h'$ is as well. But then we get $\mathbb{C} \models h'(\lambda (\varphi)) \leftrightarrow \Box_1 h'(\varphi)$, so $\Box_1$ is definable by translation.
\end{proof}

Therefore, all of Definitions \ref{def definability 1}--\ref{def definability 4} are equivalent. We may from now on say $\Box_1$ is definable in $\mathcal{L}_2$ w.r.t. $\mathbb{C}$ without worrying about what sort of definability we have in mind.

\begin{theorem}[\cite{venturiyago2020}]\label{theorem expressivity}
$\mathcal{L}_1 \preceq \mathcal{L}_2$ in $\mathbb{C}$ iff $\Box_1$ is definable in $\mathcal{L}_2$ w.r.t. $\mathbb{C}$.
\end{theorem}

We may also speak interchangeably about definability of $\Box_1$ and expressivity, in a class of frames, of its language w.r.t. another.

\begin{theorem}\label{theorem existence translation}
If $\Box_1$ is definable in $\mathcal{L}_2$ w.r.t. $\mathbb{C}$ (or equivalently, $\mathcal{L}_1 \preceq \mathcal{L}_2$), then there is a truth-preserving translation $t: \mathcal{L}_1 \to \mathcal{L}_2$ on $\mathbb{C}$.
\end{theorem}

\begin{proof}
Let $\Box_1$ be definable and $\lambda$ be the relevant mapping. Define $t: \mathcal{L}_1 \to \mathcal{L}_2$:

\begin{tabular}{lll}
$t(p)$ & $=$ & $p$\\
$t(\neg \varphi)$ & $=$ & $\neg t(\varphi)$\\
$t(\varphi \wedge \psi)$ & $=$ & $t(\varphi) \wedge t(\psi)$\\
$t(\Box_1 \varphi)$ & $=$ & $\lambda (t(\varphi))$
\end{tabular}

\noindent
One may easily check $t$ is truth-preserving.
\end{proof}

We shall, from now on, frequently make use of the above concepts, and of translations between languages. To fix notation, let us denote translations from the less expressive language to the more expressive one by $t$, and by $h$, translations which work in the opposite direction.

\section{Importing soundness}\label{sec soundness}

We start with a couple of definitions to introduce our terminology.

\begin{definition}\label{def class}
Let $P(x)$ be a property of modal structures (defined in first or second-order logic). We shall denote the class of modal structures sharing $P(x)$ by $\mathbb{C}_P$. Therefore, if $F \in \mathbb{C}_P$, then $P(F)$.
\end{definition}

\begin{definition}[Insensitivity to properties, \cite{venturiyago2020}]\label{def insensitivity property}
Let $\mathcal{L}$ be a modal language and $P$ a frame property. We say $\mathcal{L}$ is \emph{insensitive to $P(x)$} whenever there is $\mathbb{B} \supsetneq \mathbb{C}_P$, for all $\varphi \in Form_\mathcal{L}$, $\mathbb{C}_P \models \varphi$ iff $\mathbb{B} \models \varphi$.
\end{definition}

The reason we call this phenomenon \emph{insensitivity} is that no logic expressed in $\mathcal{L}$ can characterise the class of $P$-structures. If $\mathcal{L}$  is insensitive to a certain property and the largest class $\mathbb{B}$ described in Definition \ref{def insensitivity property} may be the class of all frames $\mathbb{C}_\mathbf{K}$, we say $\mathcal{L}$ is \emph{fully insensitive to $P$}.

Let $\rho$ denote an operation on structures, that is a transformation of a structure in another. Let us denote a frame $F'$ is obtained by applying $\rho$ to a frame $F$ by $F \rightsquigarrow_\rho F'$. Let us also define, for a class of structures $\mathbb{C}$, its closure under $\rho$ as $Cl_\rho(\mathbb{C}) = \{F \mid \exists F' \in \mathbb{C} (F' \rightsquigarrow_\rho F)\}$.

\begin{definition}[Collapse]\label{def collapse}
We say $\mathbb{C}_1$ and $\mathbb{C}_2$ \emph{collapse in $\mathcal{L}$} if for any $\varphi \in Form_\mathcal{L}$, $\mathbb{C}_1 \models \varphi$ iff $\mathbb{C}_2 \models \varphi$ (that is, the logics of both classes of structures in $\mathcal{L}$ are the same). 
\end{definition}

\begin{definition}[$P$-reduction]\label{def p reduction general}
Let $P(x)$ be a structure property. We call \emph{$P$-reduction} any operation $\rho$ on structures such that for any structure $F$, if $P(F)$, $F \rightsquigarrow_\rho F'$ and $F \neq F'$, then $\neg P(F')$.
\end{definition}

\begin{definition}[$P$-closure]\label{def p closure general}
Let $P(x)$ be a structure property. We call \emph{$P$-closure} any operation $\kappa$ such that, for any $F$ and $F'$, if $F \rightsquigarrow_\kappa F'$, then $P(F')$. 
\end{definition}

Notice, by our definitions, $\kappa$ being a $P$-closure does not mean $F \rightsquigarrow_\kappa F'$ implies $F \subseteq F'$; and similarly, $\rho$ being a $P$-reduction does not mean $F \rightsquigarrow_\rho F'$ implies $F' \subseteq F$. That is for we wish to use \emph{closure} and \emph{reduction} as ways of describing operations on structures which add or take away a property, respectively -- instead of, for example, in the case of relational semantics, merely adding or taking away states or accessibility arrows.

\begin{definition}[Insensitivity to structure operation]\label{def insensitivity operation}
Let $\mathcal{L}$ be a modal language, and $\rho$ an operation on structures. We say $\mathcal{L}$ is \emph{insensitive to $\rho$} if, for any $F$ and $F'$ with $F \rightsquigarrow_\rho F'$ and $\varphi \in Form_\mathcal{L}$, $F \models \varphi$ iff $F' \models \varphi$.
\end{definition}

By this definition, any modal language is insensitive to the trivial identity operation. The concept becomes fruitful when the frame operations we consider are non-trivial.

\begin{corollary}\label{corollary insensitivity operation}
If $\rho$ is a $P$-reduction and  $\mathcal{L}$ is a modal language insensitive to $\rho$, then $\mathcal{L}$ is insensitive to $P$.
\end{corollary}

\begin{proof}
Straightforward from Definitions \ref{def insensitivity property}, \ref{def p reduction general} and \ref{def insensitivity operation}.
\end{proof}

\begin{definition}[Frame recoverability]\label{def recoverable}
Let $\rho_1$ and $\rho_2$ be operations on structures, and $F$ be a structure. If for any $F'$, we have $F \rightsquigarrow_{\rho_1} F' \Rightarrow F' \rightsquigarrow_{\rho_2} F$, we say $F$ is \emph{recoverable by $\rho_2$ under $\rho_1$}.
\end{definition}

\begin{proposition}\label{proposition equivalent insensitivity}
Let $\rho_1$ and $\rho_2$ be operations and $\mathbb{C}$ be a class of structures. If $\mathcal{L}$ is insensitive to $\rho_2$, and, for any $F \in \mathbb{C}$, $F$ is recoverable by $\rho_2$ under $\rho_1$, then for any $F'$ such that $F \rightsquigarrow_{\rho_1} F'$ and $F \in \mathbb{C}$, and $\varphi \in Form_{\mathcal{L}}$, $F \models \varphi$ iff $F' \models \varphi$.
\end{proposition}

\begin{proof}
Let $F \rightsquigarrow_{\rho_1} F'$ and $\varphi \in Form_\mathcal{L}$. Suppose $F \models \varphi$. Since $F$ is recoverable by $\rho_2$ under $\rho_1$, $F' \rightsquigarrow_{\rho_2} F$; and since $\mathcal{L}$ is insensitive to $\rho_2$, that is the case iff $F' \models \varphi$.
\end{proof}

\begin{theorem}\label{theorem closure collapse}
Let $\rho$ and $\kappa$ be a $P$-reduction and $P$-closure, respectively. Let further $\mathcal{L}$ be a language insensitive to $\rho$, and $\mathbb{C}$ be a class of structures. If $Cl_\rho(\mathbb{C}_P) \subseteq \mathbb{C}$ and, for any $F \in \mathbb{C}$, $F$ is recoverable by $\rho$ under $\kappa$, then $\mathcal{L}$ collapses $\mathbb{C}$ and $\mathbb{C}_P$.
\end{theorem}

\begin{proof}
Notice $Cl_\rho(\mathbb{C}_P) \subseteq \mathbb{C}$ is equivalent to, (i) for all $F \in \mathbb{C}_P$ and $F'$, if $F \rightsquigarrow_\rho F'$, then $F' \in \mathbb{C}$. So, let $\mathbb{C} \not\models \varphi$. Then, $F \not\models \varphi$ for some $F \in \mathbb{C}$. By Definition \ref{def p closure general}, for $F \rightsquigarrow_\kappa F'$, $F' \in \mathbb{C}_P$. Since $F$ is recoverable by $\rho$ under $\kappa$, by Proposition \ref{proposition equivalent insensitivity}, $F' \not\models \varphi$, and so $\mathbb{C}_P \not\models \varphi$. Let now $\mathbb{C}_P \not\models \varphi$. Then, for some $F \in \mathbb{C}_P$, $F \not\models \varphi$. Let $F \rightsquigarrow_\rho F'$. By (i), $F' \in \mathbb{C}$, and since $\mathcal{L}$ is insensitive to $\rho$, $F' \not\models \varphi$. Thus, $\mathbb{C} \not\models \varphi$.
\end{proof}

The above theorem explains the strategy used for proving the collapse of certain classes of frames in insensitive modal logics, as in \cite{gilbertventuri2016} and \cite{venturiyago2020}. Given a frame property $P$, if a language is insensitive to a $P$-reduction $\rho$, then by showing that any frame in a certain class $\mathbb{C}$ such that $Cl_\rho(\mathbb{C}_P) \subseteq \mathbb{C}$ is recoverable by $P$-reduction under $P$-closure -- that is, that if we apply $P$-closure to a frame, the original frame is a $P$-reduction of the resulting one --, we effectively show, not only that the language is also insensitive to $P$-closure, but that the logic of $\mathbb{C}_P$ collapses to that of $\mathbb{C}$.

\begin{corollary}\label{corollary cond full insensitivity}
Let $\rho$ and $\kappa$ be a $P$-reduction and $P$-closure, respectively, and $\mathcal{L}$ be a language insensitive to $\rho$. If for any $F \in \mathbb{C}_\mathbf{K}$, $F$ is recoverable by $\rho$ under $\kappa$, then $\mathcal{L}$ is fully insensitive to $P$.
\end{corollary}

From here on, we generalize definitions and results of \cite{gilbertventuri2016}, where the strategy was first developed to offer general soundness and completeness results for \emph{RI}-logics (\emph{Reflexive-Isensitive}-logics), in the language of $\mathcal{L}^\circ$. For the completeness results, specifically, the strategy is an adaptation of that in \cite{goldblattmares2006}. In \cite{gilbertventuri2016} and \cite{venturiyago2020}, the notion of robustness w.r.t. reflexivity and seriality is defined, and it says of a class of frames that it is robust w.r.t., for example, reflexivity, if the reflexive closure of any frame of that class is still a frame of that class. However, in the present work, we find it to be more adequate to define robustness w.r.t. a certain frame operation. In the case of seriality, for example, there are many different serial closures, so that for a given class, some of them may produce frames still in that class, while others may not. Thus, it is not the property which is the most important aspect of that notion of robustness, but the specific frame operation.

\begin{definition}[Robustness w.r.t. $\rho$]\label{def robustness}
Let $\rho$ be an operation and $\mathbb{C}$ a class of frames. We say $\mathbb{C}$ is \emph{robust w.r.t. $\rho$} if $Cl_\rho(\mathbb{C}) \subseteq \mathbb{C}$.
\end{definition}

So far, we have not defined what we take a logic to be. There are two possibilities. From a proof theoretic perspective, given a certain proof theory -- Hilbert-style axiomatic system, sequent calculus, natural deduction --, a logic in a language $\mathcal{L}$ can be defined as a set $\Gamma \subseteq Form_\mathcal{L}$ closed under a certain consequence relation. Since most of the literature so far regarding insensitive logics work with axiomatic systems, a logic in this perspective may be defined as such a collection closed under the inference rules. From a semantic perspective, once we fix a semantics, and thus a notion of entailment, a logic in $\mathcal{L}$ can be defined as a set $\Gamma \subseteq Form_\mathcal{L}$ closed under this entailment relation. It is more common, however, to talk about the logic of a certain class of structures of the semantics, which is the set of formulas of $\mathcal{L}$ validated by the class. Nevertheless, the two perspectives are connected: the \emph{minimal logic of $\mathcal{L}$} is usually taken to be the logic of the class of all structures of the semantics, and while a set of formulas closed by certain inference rules which is coextensional to it is taken to be an axiomatization of the minimal logic of $\mathcal{L}$, it is sometimes equally referred to as the minimal logic. Similarly, if a certain formula characterizes a class of structures -- in the sense that a structure satisfies that formula iff it is a structure of that class --, then the closure under the minimal logic of this formula is \emph{the logic} of that class of structures. Below, we formally make the above distinction.

\begin{definition}[Logic$_{PT}$]\label{def logic pt}
Let $\mathcal{L}$ be a modal language, with a proof-theory, and an associated (logical) consequence relation $\vdash : \mathcal{P}(Form_\mathcal{L}) \to Form_\mathcal{L}$. A \emph{logic$_{PT}$} is a subset $\mathbf{L} \subseteq Form_\mathcal{L}$ closed under that consequence relation -- that is, such that $Cl_\vdash(\mathbf{L}) = \{\psi \mid \exists \Gamma \subseteq \mathbf{L} (\Gamma \vdash \psi)\} = \mathbf{L}$.\footnote{The closure operator, as usual, may be defined by progressively taking all consequences of $\mathbf{L}$, then all consequences of that set, and so on, and then collecting all the resulting sets corresponding to finite iterations of this operation.}
\end{definition}

\begin{definition}[Logic$_S$]\label{def logic pt}
Let $\mathcal{L}$ be a modal language, with a structure-based semantics, and an associated entailment relation $\models$ such that $\Gamma \models \psi$ iff for any structure $F$, $F \models \Gamma \Rightarrow F \models \psi$. A \emph{logic$_S$} is a subset $\mathbf{L} \subseteq Form_\mathcal{L}$ closed under that entailment relation -- that is, $Cl_{\models}(\mathbf{L}) = \{\psi \mid \exists \Gamma \subseteq \mathbf{L} (\Gamma \models \psi)\} = \mathbf{L}$.
\end{definition}

From now on, unless stated otherwise, when referring to a logic $\mathbf{L}$, we shall consider $\mathbf{L}$ to be a logic$_S$. Moreover, given a certain semantic, we say $\mathbf{L}$ is \emph{consistent} w.r.t. it if there is a structure validating all the formulas of $\mathbf{L}$. If the semantics in question is clear by context, we simply say the logic is consistent. Furthermore, we denote by $\mathbb{C}_\mathbf{L}$ the class of structures which validates all formulas of $\mathbf{L}$. Notice, in the limit case $\mathbf{L}$ is inconsistent, $\mathbb{C}_\mathbf{L} = \varnothing$.

Now, fix a semantics. Let $\mathcal{L}_1$ and $\mathcal{L}_2$ be modal languages such that $\mathcal{L}_2 \preceq \mathcal{L}_1$, $\mathbf{L}$ a consistent logic in $\mathcal{L}_1$, and $h: \mathcal{L}_1 \to \mathcal{L}_2$ be truth-preserving translation in $\mathbb{C}_\mathbf{L}$. Let also $\mathbf{L}_h$ be the smallest logic of $\mathcal{L}_2$ such that $h[\mathbf{L}] = \{h(\varphi) \mid \varphi \in \mathbf{L}\} \subseteq \mathbf{L}_h$ -- that is, $\mathbf{L}_h = Cl_{\models}(h[\mathbf{L}])$. Then, we have the following result.

\begin{theorem}\label{theorem soundness 1}
Let $\mathcal{L}_2 \preceq \mathcal{L}_1$ and $\mathbf{L}$ be a logic in $\mathcal{L}_1$. Then, $\mathbf{L}_h$ is sound w.r.t. $\mathbb{C}_\mathbf{L}$.
\end{theorem}

\begin{proof}
Notice $Cl_{\models}(h[\mathbf{L}])$ is constructed in the following way: defining $\mathbf{L}_0 = h[\mathbf{L}]$ and $\mathbf{L}_{n+1} = \{\varphi \mid \exists \Gamma \subseteq \mathbf{L}_n (\Gamma \models \varphi)\}$, and then letting $Cl_{\models}(h[\mathbf{L}]) = \bigcup_{n \in \omega} \mathbf{L}_n$. Let now $\varphi \in \mathbf{L}_h$. We show, by induction on the first level of the construction of $Cl_{\models}(h[\mathbf{L}]) = \mathbf{L}_h$ in which $\varphi$ appears, that $\mathbb{C}_\mathbf{L} \models \varphi$. In case $\varphi \in \mathbf{L}_0$, then $\varphi = h(\psi)$ for some $\psi \in \mathbf{L}$, so since $h$ is truth-preserving in $\mathbb{C}_\mathbf{L}$, $\mathbb{C}_\mathbf{L} \models \psi$. Let now $\varphi \in \mathbf{L}_{n+1}$. Then, there is $\Gamma \subseteq \mathbf{L}_n$ such that $\Gamma \models \varphi$ in $\mathbb{C}_\mathbf{L}$. By induction hypothesis, $\mathbb{C}_\mathbf{L} \models \gamma$ for every $\gamma \in \Gamma$, and therefore, $\mathbb{C}_\mathbf{L} \models \varphi$.
\end{proof}

\begin{theorem}\label{theorem soundness insensitivity}
Let $\mathcal{L}_2 \preceq \mathcal{L}_1$ and $\mathbf{L}$ be a logic in $\mathcal{L}_1$. If $\mathcal{L}_2$ is insensitive to $\rho$, then $\mathbf{L}_h$ is sound w.r.t. $Cl_\rho(\mathbb{C}_\mathbf{L})$.
\end{theorem}

\begin{proof}
Let $\varphi \in \mathcal{L}_2$ and $Cl_\rho(\mathbb{C}_\mathbf{L}) \not\models \varphi$. Then, for some $F \in Cl_\rho(\mathbb{C}_\mathbf{L})$, $F \not\models \varphi$. By definition, we have for some $F' \in \mathbb{C}_\mathbf{L}$, $F' \rightsquigarrow_\rho F$. Since $\mathcal{L}_2$ is insensitive to $\rho$, we have $F' \not\models \varphi$, so by Theorem \ref{theorem soundness 1}, $\varphi \not\in \mathbf{L}_h$.
\end{proof}

\begin{corollary}\label{corollary soundness}
Let $\mathcal{L}_2 \preceq \mathcal{L}_1$ and $\mathbf{L}$ be a logic in $\mathcal{L}_1$. If $\mathcal{L}_2$ is insensitive to $\rho$, then $\mathcal{L}_2$ collapses $\mathbb{C}_\mathbf{L}$ and $Cl_\rho(\mathbb{C}_\mathbf{L})$.
\end{corollary}

\begin{proof}
By Definition \ref{def collapse} and Theorems \ref{theorem soundness 1} and \ref{theorem soundness insensitivity}.
\end{proof}

Therefore, for any modal language $\mathcal{L}_2$, if we have an appropriate translation $h$ from a more expressive modal language $\mathcal{L}_1$, and an insensitivity of $\mathcal{L}_2$ to an operation on structures $\rho$, then for any logic $\mathbf{L}$ of $\mathcal{L}_1$, we obtain the soundness of $\mathbf{L}_h$ w.r.t. $\mathbb{C}_\mathbf{L}$ -- and even more restrictively, $Cl_\rho(\mathbb{C}_\mathbf{L})$. Notice, there is no mention of what sort of semantics we are using. The above soundness criterion applies to, for example, topological or neighbourhood semantics. A strategy therefore is presented for proving soundness of a diverse class of logics in different languages: let $\mathcal{L}_1$ be a modal language, and $\mathbf{L}$ a logic of it. Then, we have \emph{Soundness strategy $1$}:

\begin{enumerate}
\item let $\mathcal{L}_2$ be a less expressive language;
\item define a translation $h: \mathcal{L}_1 \to \mathcal{L}_2$ truth-preserving in $\mathbb{C}_\mathbf{L}$;
\item then, by Theorem \ref{theorem soundness 1}, we have the soundness of $\mathbf{L}_h = Cl_{\models}(h[\mathbf{L}])$ w.r.t. $\mathbb{C}_\mathbf{L}$; and if there is an operation on frames $\rho$, to which $\mathcal{L}_2$ is insensitive, by Theorem \ref{theorem soundness insensitivity}, we have the soundness of $\mathbf{L}_h$ w.r.t. $Cl_\rho(\mathbb{C}_\mathbf{L})$.
\end{enumerate}

This method generalizes the method used in \cite{gilbertventuri2016} for the soundness of RI-logics. The above results are built upon the assumption we have a truth-preserving translation from the more expressive language to the less expressive one. Not always is that the case. What we do always have, however, is a truth-preserving translation $t$ on the other direction, by the definability of one operator in terms of the other.

Let us try the opposite: suppose $\mathbf{L}$ is a logic in $\mathcal{L}_2$, with $\mathcal{L}_2 \preceq \mathcal{L}_1$, and similarly to the definition of $\mathbf{L}_h$, let $\mathbf{L}_t$ be the smallest logic of $\mathcal{L}_1$ containing the image of $\mathbf{L}$ under the truth-preserving translation $t: \mathcal{L}_2 \to \mathcal{L}_1$ -- guaranteed by Theorem \ref{theorem existence translation}.

For a logic $\mathbf{L}$ of $\mathcal{L}_2$ and another $\mathbf{L}'$ of $\mathcal{L}_1$, let a mapping $t : \mathcal{L}_2 \to \mathcal{L}_1$ be \emph{theorem preserving} from $\mathbf{L}$ to $\mathbf{L}'$ if $\vdash^\mathbf{L} \varphi$ implies $\vdash^{\mathbf{L}'} t(\varphi)$ -- that is, if the fact $\varphi$ is a theorem of $\mathbf{L}$ implies its translation is a theorem of $\mathbf{L}'$.

\begin{theorem}\label{theorem soundness 2}
Let $t$ be theorem preserving from $\mathbf{L}$ to $\mathbf{L}_t$ and truth-preserving in $\mathbb{C}_{\mathbf{L}_t}$. Then $\mathbf{L}$ is sound w.r.t. $\mathbb{C}_{\mathbf{L}_t}$.
\end{theorem}

\begin{proof}
Let $\varphi \in Form_{\mathcal{L}_2}$ and $\vdash^\mathbf{L} \varphi$. Since $t$ is theorem preserving, $\vdash^{\mathbf{L}_t} t(\varphi)$. By the definition of $\mathbb{C}_{\mathbf{L}_t}$, that means $\mathbb{C}_{\mathbf{L}_t} \models t(\varphi)$. Since $t$ is truth-preserving in that class of frames, $\mathbb{C}_{\mathbf{L}_t} \models \varphi$.
\end{proof}

\begin{theorem}\label{theorem soundness 3}
If $\mathcal{L}_2$ is insensitive to $\rho$, then $\mathbf{L}$ is sound w.r.t. $Cl_\rho(\mathbb{C}_{\mathbf{L}_t})$.
\end{theorem}

\begin{proof}
Let $\varphi \in Form_{\mathcal{L}_2}$ and $Cl_\rho(\mathbb{C}_{\mathbf{L}_t}) \not\models \varphi$. Then, for some $F \in Cl_\rho(\mathbb{C}_{\mathbf{L}_t})$, $F \not\models \varphi$. We have $F' \rightsquigarrow_\rho F$ for some $F' \in \mathbb{C}_{\mathbf{L}_t}$. Since $\mathcal{L}_2$ is insensitive to $\rho$, $F' \not\models \varphi$, and since $t$ is truth-preserving, $F' \not\models t(\varphi)$. Thus, by the definition of $\mathbf{L}_t$, that means $t(\varphi) \not\in \mathbf{L}_t$. Therefore, $\varphi \not\in \mathbf{L}$.
\end{proof}

We thus have \emph{Soundness strategy $2$}. Let $\mathcal{L}_2$ be a modal language, and $\mathbf{L}$ a logic of it. Then,

\begin{enumerate}
\item let $\mathcal{L}_1$ be a more expressive language;
\item define a translation $t: \mathcal{L}_2 \to \mathcal{L}_1$ truth-preserving in $\mathbb{C}_{\mathbf{L}_t}$ and theorem preserving from $\mathbf{L}$ to $\mathbf{L}_t$;
\item then, by Theorem \ref{theorem soundness 2}, we have the soundness of $\mathbf{L}$ w.r.t. $\mathbb{C}_{\mathbf{L}_t}$; and if there is an operation on frames $\rho$, to which $\mathcal{L}_2$ is insensitive, by Theorem \ref{theorem soundness 3}, we have the soundness of $\mathbf{L}$ w.r.t. $Cl_\rho(\mathbb{C}_{\mathbf{L}_t})$.
\end{enumerate}

This strategy may seem simpler than the previous one, since the existence of a truth-preserving translation is guaranteed. Furthermore, in case the more expressive language is the normal modal language, it allows the use of its well-developed correspondence theory to fund the classes of frames w.r.t. which the logics of the less expressive language are sound. However, the translation must also be theorem preserving, something which is not always given.

Whereas we have considered what the notion of logic$_S$ gives us, we now investigate how similar results may be given for a logic$_{PT}$. For the remainder of this section, whenever talking about a logic, we mean a logic$_{PT}$. Now, we need to assume we have a proof theory. Suppose we have an axiomatization for the minimal logic in $\mathcal{L}_2$, and let $\mathbf{L}$ be a logic in $\mathcal{L}_1$. Define now $\mathbf{L}_h$ as the closure of $h[\mathbf{L}]$ under the inference rules of the minimal logic of $\mathcal{L}_2$. Then:

\begin{theorem}\label{theorem soundness insensitivity 2}
Let $\mathcal{L}_2 \preceq \mathcal{L}_1$ and $\mathbf{L}$ be a logic in $\mathcal{L}_1$. If $\mathcal{L}_2$ is insensitive to $\rho$, then $\mathbf{L}_h$ is sound w.r.t. $Cl_\rho(\mathbb{C}_\mathbf{L})$.
\end{theorem}

\begin{proof}
Let $\varphi \in Form_{\mathcal{L}_2}$ and $Cl_\rho(\mathbb{C}_\mathbf{L}) \not\models \varphi$. Then, for some $F \in Cl_\rho(\mathbb{C}_\mathbf{L})$, $F \not\models \varphi$. Thus, for some $F' \in \mathbb{C}_\mathbf{L}$, $F' \rightsquigarrow_\rho F$. By the insensitivity, $F' \not\models \varphi$. Then, we have two cases. If for some $\psi \in Form_{\mathcal{L}_1}$, $\varphi = h(\psi)$, then since $h$ is truth-preserving, $F' \not\models \psi$, so by the soundness of $\mathbf{L}$, $\psi \not\in \mathbf{L}$. Suppose $h(\psi) \in \mathbf{L}_h$. Then, by the definition of $\mathbf{L}_h$, $h(\psi)$ must be obtained by the inference rules of the minimal logic of $\mathcal{L}_2$, when applied to a set of translations of theorems of $\mathbf{L}$. Let $h[\Gamma]$ be the translated theorems from which $h(\psi)$ follows. By definition, the minimal logic is the logic of the class of all structures, and thus its inference rules preserve truth. That means $h[\Gamma] \models h(\psi)$. By definition, since $F' \in \mathbb{C}_\mathbf{L}$, $F' \models \Gamma$. Since $h$ is truth-preserving, $F' \models h[\Gamma]$, which means $F' \models h(\psi)$, a contradiction. Therefore, $h(\psi) = \varphi \not\in \mathbf{L}_h$. If now for all $\psi \in \mathbf{L}$, $\varphi \neq h[\psi]$, then once again $\varphi$ must be obtained by the inference rules of the minimal logic of $\mathcal{L}_2$ in the same way as previously described. In that case, the previous argument equally follows.
\end{proof}

Similarly to the semantic perspective, we may consider the direction of the translation easily obtainable, $t : \mathcal{L}_2 \to \mathcal{L}_1$. In that case, we assume we have an axiomatization for the minimal logic in $\mathcal{L}_1$, and let $\mathbf{L}$ be a logic in $\mathcal{L}_2$. We then define $\mathbf{L}_t$ as the closure of $t(\mathbf{L})$ under the inference rules of the minimal logic of $\mathcal{L}_1$. Now:

\begin{theorem}\label{theorem soundness 4}
If $\mathcal{L}_2$ is insensitive to $\rho$, then $\mathbf{L}$ is sound w.r.t. $Cl_\rho(\mathbb{C}_{\mathbf{L}_t})$.
\end{theorem}

\begin{proof}
Let $\varphi \in Form_{\mathcal{L}_2}$ and $Cl_\rho(\mathbb{C}_{\mathbf{L}_t}) \not\models \varphi$. Then, for some $F \in Cl_\rho(\mathbb{C}_{\mathbf{L}_t})$, $F \not\models \varphi$. We have $F' \rightsquigarrow_\rho F$ for some $F' \in \mathbb{C}_{\mathbf{L}_t}$, so since $\mathcal{L}_2$ is insensitive to $\rho$, $F' \not\models \varphi$. Since $t$ is truth-preserving, $F' \not\models t(\varphi)$. Suppose $t(\varphi) \in \mathbf{L}_t$. Analogously to the previous theorem, by the definition of $\mathbf{L}_t$, $t(\varphi)$ must be obtained by the inference rules of the minimal logic of $\mathcal{L}_1$, when applied to a set of translations of theorems of $\mathbf{L}$, so let $t[\Gamma]$ be the translated theorems from which $t[\varphi]$ follows. We have $t[\Gamma] \models t(\varphi)$. By definition, since $F' \in \mathbb{C}_{\mathbf{L}_t}$, $F' \models t[\Gamma]$, which means $F' \models t(\varphi)$, a contradiction. Therefore, $t(\varphi) \not\in \mathbf{L}_t$, and therefore $\varphi \not\in \mathbf{L}$.
\end{proof}

So we find, in both sorts of definition of what a logic is, a way for insensitivity to help define the class of frames which validate it.

\section{Importing completeness in relational frames}\label{sec completeness}

We now present a method for showing completeness which is more restricted than the former one for soundness. It is restricted for two reasons: it only covers modal logics with a single primitive unary operator, and w.r.t. relational semantics. For clarity, let us denote a logic in a language with a modal operator $\lhd$ by a superscript (e.g. $\mathbf{L}^\lhd$), and a logic in normal modal language by a letter, simply (e.g. $\mathbf{L}$). For the next results, we generalise the following usual concepts in proofs by canonical model for modal logics.

\begin{definition}[Canonical model]\label{def canonical model}
Let $\mathbf{L}$ be a modal logic. We call a model $\mathcal{M} = \langle W, ..., V \rangle$ \emph{canonical for} $\mathbf{L}$ if:

\begin{enumerate}
\item $W$ is the set of all maximal $\mathbf{L}$-consistent sets of formulas of $\mathcal{L}$;
\item for every propositional variable $p$, $V(p) = \{x \in W \mid p \in X\}$;
\item for any $\varphi \in Form_\mathcal{L}$ and $w \in W$, $\mathcal{M}, w \models \varphi$ iff $\varphi \in w$.
\end{enumerate}
\end{definition}

\begin{definition}[Canonical logic]\label{def canonical logic}
Let $\mathbf{L}$ be a modal logic and $\mathcal{M}$ its canonical model. Let $F$ be the frame on which it is based, and $\mathbb{C}_\mathbf{L}$ the class of frames determined by $\mathbf{L}$. We call $\mathbf{L}$ \emph{canonical} if $F \in \mathbb{C}_\mathbf{L}$.
\end{definition}

We leave the above definitions open enough so they may encompass, for example, canonical neighborhood models and canonical logics for which there are no relational frames, despite the present use in relational semantics. For our presentation, we shall consider an arbitrary modal language $\mathcal{L}^\lhd$ with a single primitive modal operator $\lhd$.

\begin{definition}[Boxlike mapping]
Let $\mathbf{L}^\lhd$ be a logic of $\mathcal{L}^\lhd$, $\langle W_{\mathbf{L}^\lhd}, R_{\mathbf{L}^\lhd} \rangle$ its canonical (relational) frame, and $h : \mathcal{L}^\Box \to \mathcal{L}^\lhd$ a translation. We say $h$ is \emph{boxlike} w.r.t. $\mathbf{L}^\lhd$ if for any $x, y \in W_{\mathbf{L}^\lhd}$, $y \in R_{\mathbf{L}^\lhd}(x) \Leftrightarrow \forall \varphi \in Form_{\mathcal{L}^\Box} (h(\Box \varphi) \in x \Rightarrow h(\varphi) \in y)$.
\end{definition}

Suppose we have Lindenbaum's lemma for $\mathcal{L}^\lhd$ -- which we may as long as the size of $Form_{\mathcal{L}^\lhd}$ is countable. Define $\mathbf{L}^\lhd_h$ as the smallest logic containing the image of the normal modal logic $\mathbf{L}$ under $h$ -- that is, $\mathbf{L}^\lhd_h$ must be closed under its inference rules, which we take to be given --, and $F_{\mathbf{L}^\lhd_h} = \langle W_{\mathbf{L}^\lhd_h}, R_{\mathbf{L}^\lhd_h} \rangle$ be its canonical frame. We shall proceed with the supposition $\mathbf{L}$ is a consistent logic (otherwise it is trivially complete). We also define the mapping $f: W_{\mathbf{L}^\lhd_h} \to W_\mathbf{L}; x \mapsto \{\varphi \mid h(\varphi) \in x\}$ -- which, as we shall shortly see, is well defined.

We recall, in a canonical model in the normal modal language, $xRy$ iff $L(x) \subseteq y$, where $L(x) = \{\varphi \mid \Box \varphi \in x\}$.

A mapping $h$ is compositional for non-modal formulas if it has the following behaviour:

\begin{tabular}{lll}
$h(p)$ & $=$ & $p'$\\
$h(\neg \varphi)$ & $=$ & $\neg h(\varphi)$\\
$h(\varphi \wedge \psi)$ & $=$ & $h(\varphi) \wedge h(\psi)$ \end{tabular}

\noindent
It is worth to note that in the results we here seek to generalize, for example, \cite{gilbertventuri2016} defines a translation from $\mathcal{L}^\Box$ to $\mathcal{L}^\circ$ which is compositional, truth-preserving and theorem preserving.\footnote{For example, the operator $\circ$ is semantically defined to have the same truth condition as $\neg \varphi \vee \Box \varphi$. \cite{gilbertventuri2016} define a translation $* : \mathcal{L}^\Box \to \mathcal{L}^\circ$ recursively given by: $*(p) = p$, $*(\neg \varphi) = \neg *(\varphi)$, $*(\varphi \wedge \psi) = *(\varphi) \wedge *(\psi)$, and $*(\Box \varphi) = *(\varphi) \wedge \circ *(\varphi)$. This mapping is truth preserving in $\mathbb{C}_\mathbf{K}$, compositional and theorem preserving. Of course, preservation of theoremhood, in general, will depend on the specific axiomatization of the relevant non-normal logic.}

\begin{lemma}\label{lemma f is meaningful}
Let $h$ be compositional and theorem preserving from $\mathbf{L}$ to $\mathbf{L}^\lhd_h$. For any $x \in W_{\mathbf{L}^\lhd_h}$, $f(x) \in W_\mathbf{L}$ (that is, $f$ is well defined).
\end{lemma}

\begin{proof}
Assume $f(x)$ is not $\mathbf{L}$-consistent. Then there are $\alpha_i \in f(x)$ such that \linebreak $\vdash^\mathbf{L} \bigwedge_{1 \leq i \leq n} \alpha_i \rightarrow \bot$. Since $h$ preserves theoremhood from $\mathbf{L}$ to $\mathbf{L}^\lhd_h$ and is compositional, $\vdash^{\mathbf{L}^\lhd_h} \bigwedge_{1 \leq i \leq n} h(\alpha_i) \rightarrow \bot$, and by the definition of $f(x)$, $h(\alpha_i) \in x$, which means $x$ must be $\mathbf{L}^\lhd_h$-inconsistent. Assume now $f(x)$ is not maximal. Then, there is $\alpha \in Form_{\mathcal{L}^\Box}$ such that $\alpha, \neg \alpha \not\in f(x)$. By the definition of $f(x)$, then, $h(\alpha), h(\neg \alpha) \not\in x$. By the compositionality of $h$, we get $h(\neg \alpha) = \neg h(\alpha) \in x$, and so $h(\alpha), \neg h(\alpha) \not\in x$, which means $x$ itself is not maximal.
\end{proof}

Let us say $h$ is theorem reflecting from $\mathbf{L}$ to $\mathbf{L}'$ if we have the converse of theorem preservation, that is, if $\vdash^{\mathbf{L}'} h(\varphi)$ implies $\vdash^\mathbf{L} \varphi$.

\begin{lemma}\label{lemma theorem reflect}
Let $\mathbf{L}$ be complete w.r.t. $\mathbb{C}_\mathbf{L}$. If $h$ is truth-preserving in $\mathbb{C}_\mathbf{L}$, and $\mathbf{L}^\lhd_h$ is sound w.r.t. $\mathbb{C}_\mathbf{L}$, then $h$ is theorem reflecting from $\mathbf{L}$ to $\mathbf{L}^\lhd_h$.\footnote{Notice being truth-preserving in $\mathbb{C}_\mathbf{L}$ is a condition weaker than simple truth preservation, which entails preservation in the class of all frames.}
\end{lemma}

\begin{proof}
Let $\varphi \in Form_{\mathcal{L}^\Box}$. Since $\mathbf{L}^\lhd_h$ is sound w.r.t. $\mathbb{C}_\mathbf{L}$, if $\vdash^{\mathbf{L}^\lhd_h} h(\varphi)$ then $\mathbb{C}_\mathbf{L} \models h(\varphi)$. Since $h$ is truth-preserving in $\mathbb{C}_\mathbf{L}$, that means $\mathbb{C}_\mathbf{L} \models \varphi$, but since $\mathbf{L}$ is complete w.r.t. $\mathbb{C}_\mathbf{L}$, we finally get $\vdash^\mathbf{L} \varphi$.
\end{proof}

\begin{definition}[Surjectivity up to equivalence]
Let $h : \mathcal{L}_1 \to \mathcal{L}_2$ be a mapping and $\mathbf{L}$ a logic in $\mathcal{L}_1$. We say $h$ is \emph{surjective up to equivalence w.r.t. $\mathbf{L}$} if for any $\varphi \in Form_{\mathcal{L}_2}$ there is $\psi \in Form_{\mathcal{L}_1}$ such that $\vdash^{\mathbf{L}_h} \varphi \leftrightarrow h(\psi)$.
\end{definition}

\begin{lemma}\label{lemma f injective}
Let $h$ be compositional, boxlike w.r.t. $\mathbf{L}^\lhd_h$ and surjective up to equivalence w.r.t. $\mathbf{L}$. Then $f$ is injective, and, if $xR_{\mathbf{L}^\lhd_h}y$, then $f(x)R_\mathbf{L}f(y)$.
\end{lemma}

\begin{proof}
For injectivity, let $x \neq y$. With no loss of generality, we may assume there is $\varphi \in x \setminus y$. Since $h$ is surjective up to equivalence w.r.t. $\mathbf{L}$, $\vdash^{\mathbf{L}^\lhd_h} \varphi \leftrightarrow h(\psi)$ for some $\psi \in Form_{\mathcal{L}^\Box}$, so that $h(\psi) \in x$, and by the definition of $f$, $\psi \in f(x)$. By the maximality of $y$, we have $\neg h(\psi) \in y$. By the compositionality of $h$, $\neg h(\psi) = h(\neg \psi)$, so that $\neg \psi \in f(y)$. By the consistency of $f(y)$, $\neg \psi \not\in f(y)$, thus showing that $f(x) \neq f(y)$.

For the second claim, assume $\varphi \in L(f(x))$, and so $\Box \varphi \in f(x)$. Then, by the definition of $f(x)$, $h(\Box \varphi) \in x$. Since $h$ is boxlike w.r.t. $\mathbf{L}^\lhd_h$, that means $h(\varphi) \in y$. Hence, $\varphi \in f(y)$, by the definition of $f$.
\end{proof}

\begin{lemma}\label{lemma f preserves R}
Let $h$ be compositional and truth-preserving in $\mathbb{C}_\mathbf{L}$. Let also $\mathbf{L}$ be complete w.r.t. $\mathbb{C}_\mathbf{L}$, and $\mathbf{L}^\lhd_h$ be sound w.r.t. $\mathbb{C}_\mathbf{L}$. If $f(x)R_\mathbf{L}z$, then there is $y \in W_{\mathbf{L}^\lhd_h}$ such that $xR_{\mathbf{L}^\lhd_h}y$ and $f(y) = z$.
\end{lemma}

\begin{proof}
Let $f(x) R_\mathbf{L} z$, which means $L(f(x)) \subseteq z$, so that $\{\varphi \mid h(\varphi) \in x\} \subseteq z$. Let 

$$y_0 = \{h(\varphi) \mid h(\Box \varphi) \in x\} \cup \{h(\psi) \mid \psi \in z\}$$

\noindent
Suppose $y_0$ is inconsistent, and so for some $h(\varphi), h(\psi) \in y_0$, $\vdash^{\mathbf{L}^\lhd_h} h(\varphi) \rightarrow \neg h(\psi)$. By Lemma \ref{lemma theorem reflect}, $h$ is theorem reflecting from $\mathbf{L}$ to $\mathbf{L}^\lhd_h$, so since $h$ is also compositional, $\vdash^\mathbf{L} \varphi \rightarrow \neg \psi$, which means $\varphi \rightarrow \neg \psi \in z$. Since $z$ is consistent, either $\varphi$ or $\psi \not\in z$, so suppose the former. By the construction of $y_0$, that means $h(\Box \varphi) \in x$, and so $\Box \varphi \in f(x)$, which means $\varphi \in L(f(x)) \subseteq z$, a contradiction. Suppose then the latter. Again, that means $h(\Box \psi) \in x$, so that $\Box \psi \in f(x)$, and thus $\psi \in z$, a contradiction once again.

Let us now, by Lindenbaum's lemma, extend $y_0$ to a maximal consistent $y$. Let $\varphi \in \mathcal{L}^\Box$, and $h(\Box \varphi) \in x$. Then, by construction, $h(\varphi) \in y$. Therefore, by the arbitrariness of $\varphi$ and the definition of $h$, $y \in R_{\mathbf{L}^\lhd_h}(x)$. Let now $\varphi \in z$. Then, by construction, $h(\varphi) \in y$, and therefore $\varphi \in f(y)$. Since $z$ is maximal, that means $z = f(y)$.
\end{proof}

We recall the following definitions and results (\citep{blackburn2001}).

\begin{definition}[Bounded morphism]\label{def bounded}
Let $F_1 = \langle W_1, R_1 \rangle$ and $F_2 = \langle W_2, R_2 \rangle$ be frames. Then, $f : W_1 \to W_2$ is a \emph{bounded morphism} from $F_1$ to $F_2$ when the following two conditions are met: $x R_1 y$ implies $f(x) R_2 f(y)$; if $f(x) R_2 z$, then there is $y$ such that $x R_1 y$ and $f(y) = z$. When there is a surjective bounded morphism from $F_1$ onto $F_2$, we say to $F_2$ is a \emph{bounded morphic image} of $F_1$ ($F_1 \twoheadrightarrow F_2$).
\end{definition}

\begin{definition}[Generated subframe]\label{def gen subframe}
Let $F_1 = \langle W_1, R_1 \rangle$ and $F_2 = \langle W_2, R_2 \rangle$ be frames. $F_2$ is a \emph{generated subframe} of $F_1$ ($F_2 \rightarrowtail F_1$) when $F_2$ is a subframe of $F_1$ and if $x \in W_2$ and $x R_1 y$, then $y \in W_2$.
\end{definition}

\begin{theorem}\label{theorem completeness 1}
Let $h$ be compositional, theorem preserving from $\mathbf{L}$ to $\mathbf{L}^\lhd_h$, truth preserving in $\mathbb{C}_\mathbf{L}$, boxlike w.r.t. $\mathbf{L}^\lhd_h$, and surjective up to equivalence w.r.t. $\mathbf{L}$. Let also $\mathbf{L}$ be a canonical normal modal logic which is complete w.r.t. $\mathbb{C}_\mathbf{L}$, and $\mathbf{L}^\lhd_h$ be sound w.r.t. $\mathbb{C}_\mathbf{L}$. Then, $\mathbf{L}_h^\lhd$ is complete w.r.t. $\mathbb{C}_\mathbf{L}$.
\end{theorem}

\begin{proof}
Following the above definitions and Lemmas \ref{lemma f is meaningful}--\ref{lemma f preserves R} we have that: $f$ is an injective bounded morphism from $F_{\mathbf{L}^\lhd_h}$ to $F_\mathbf{L}$; $f[F_{\mathbf{L}^\lhd_h}]$ (the image of $F_{\mathbf{L}^\lhd_h}$ under $f$) is a generated subframe of $F_\mathbf{L}$; and $f[F_{\mathbf{L}^\lhd_h}]$ is isomorphic to $F_{\mathbf{L}_h}$. Therefore, $F_{\mathbf{L}^\lhd_h}$ is a frame for $\mathbf{L}$. Suppose $\not\vdash^{\mathbf{L}^\lhd_h} \varphi$. Since $h$ is surjective up to equivalence, there is $\psi \in Form_{\mathcal{L}^\Box}$ such that $\vdash^{\mathbf{L}^\lhd_h} \varphi \leftrightarrow h(\psi)$, and thus $\not\vdash^{\mathbf{L}^\lhd_h} h(\psi)$. Thus, there is $w \in W_{\mathbf{L}^\lhd_h}$ such that $\neg h(\psi) \in w$. Since $f[F_{\mathbf{L}^\lhd_h}]$ is isomorphic to $F_{\mathbf{L}_h}$ and $h$ is compositional, by the definition of $f$ that means $\neg \psi \in f(w)$. Thus, $F_\mathbf{L} \not\models \psi$. Since $\mathbf{L}$ is canonical, that means $\mathbb{C}_\mathbf{L} \not\models \psi$, so since $h$ is truth-preserving in $\mathbb{C}_\mathbf{L}$, $\mathbb{C}_\mathbf{L} \not\models h(\psi)$. Since $\mathbf{L}^\lhd_h$ is sound w.r.t. $\mathbb{C}_\mathbf{L}$, $\mathbb{C}_\mathbf{L} \models \varphi \leftrightarrow h(\psi)$, and therefore $\mathbb{C}_\mathbf{L} \not\models \varphi$.
\end{proof}

Therefore, we have \emph{Completeness strategy $1$} of $\mathbf{L}^\lhd_h$ w.r.t. $\mathbb{C}_\mathbf{L}$ by:

\begin{enumerate}
\item defining a compositional mapping $h: \mathcal{L}^\Box \to \mathcal{L}^\lhd$ which is theorem preserving from $\mathbf{L}$ to $\mathbf{L}^\lhd_h$, surjective up to equivalence w.r.t. $\mathbf{L}$, truth-preserving in $\mathbb{C}_\mathbf{L}$, and boxlike w.r.t. $\mathbf{L}^\lhd_h$;
\item showing $\mathbf{L}$ is canonical and complete w.r.t. $\mathbb{C}_\mathbf{L}$;
\item showing $\mathbf{L}^\lhd_h$ is sound w.r.t. $\mathbb{C}_\mathbf{L}$.
\end{enumerate}

We note that the truth-preservation of $h$ and the soundness of $\mathbf{L}_h$ w.r.t. $\mathbb{C}_\mathbf{L}$ is needed because the proof of Lemma \ref{lemma f preserves R} requires $h$ to reflect theoremhood. In case $h$ is already theorem reflecting, the proof of Lemma \ref{lemma f preserves R} and Theorem \ref{theorem completeness 1} may be carried out without the assumption of its truth-preservation, and of the soundness of $\mathbf{L}^\lhd_h$.

Just as in the case of soundness, we may now address a strategy following from the existence of a translation in the opposite direction, that is, a translation $t: \mathcal{L}^\lhd \to \mathcal{L}^\Box$. First, let us recall how proofs by canonical models, in relational semantics for normal modal logic $\mathbf{L}$, go. First, one defines a set $W$ of all maximal consistent (in the logic of interest) sets of formulas. Then, one defines the accessibility relation $R$ in such a way that $xRy$ iff $L(x) \subseteq y$. Then, one adds to that frame a canonical valuation $V(p) = \{x \in W \mid p \in x\}$, and shows the Main Lemma -- that the resulting model $\mathcal{M}$ is such that, for any $x \in W$ and $\varphi \in Form_{\mathcal{L}^\Box}$, $\mathcal{M}, x \models \varphi$ iff $\varphi \in x$. At the last step, one shows the frame of this canonical model is in the class of frames determined by the logic, and that therefore, any non-theorem of $\mathbf{L}$ is invalidated by the canonical model, and thus by that class of frames.

As it may be seen in \cite{gilbertventuri2016}, \cite{humberstone1995}, \cite{steinsvold2011}, \cite{gilbertkubyshkina2021}, and \cite{yagoventuri2026}, offering a canonical model for the minimal logic of a modal language, in relational semantics, requires an adequate definition of the condition by which the accessibility relation should behave -- a condition which needs to ensure, should it be satisfied, that any formula which is in some sense necessary, belonging to a world $x$, needs to be in any world accessed by $x$. That may be a hard task, and such a definition is not always unique. The accessibility relation of $R$ of a canonical model is, analogously to the modal case, given by $xRy$ iff $L^\lhd(x) \subseteq y$, where $L^\lhd(x)$ is some set of formulas derivative of the modal formulas in $x$. The difficulty, of course, lies precisely in spelling out just what this set is. For example, in \cite{humberstone1995}, we find $L^\Delta(x) = \{\varphi \mid \Delta\ \varphi \in x\ \&\ \text{if}\ \vdash^{\mathbf{B}^\rho_\mathbf{K}} \varphi \rightarrow \psi,\ \text{then}\ \Delta \psi \in x\}$. In \cite{marcos2005} and \cite{gilbertventuri2016} we find $L^\circ(x) = \{\varphi \mid \text{for any}\ \psi,\ \circ(\varphi \wedge \psi) \in x\}$, and $L^\circ(x) = \{\varphi \mid \varphi \wedge \circ \varphi \in x\}$, respectively. A way of avoiding this explicit definition may be useful.

\begin{theorem}[Proof by canonical model]\label{theorem canonical model}
Let $\mathbf{L} \supseteq \mathbf{K}$ be a normal modal logic, and $\mathcal{M}_\mathbf{K} = \langle W_\mathbf{K}, R_\mathbf{K}, V_\mathbf{K} \rangle$ be the canonical model of $\mathbf{K}$. Define $\mathcal{M}' = \langle W', R', V_\mathbf{K} \rangle$, such that $W' = \{x \mid  x\ \text{is maximal}\ \mathbf{L} \text{-consistent}\}$ and $R' = R_\mathbf{K} \restriction_{W'}$. Then, $\mathcal{M}'$ is a canonical model for $\mathbf{L}$.
\end{theorem}

\begin{proof}
First, notice $W' \subseteq W_\mathbf{K}$, since any $\mathbf{L}$-consistent set of formulas must be $\mathbf{K}$-consistent. Now, suppose $\vdash^\mathbf{L} \varphi$. Then, for any $x \in W'$, $\varphi \in x$. For any non-modal formula, clearly, $\mathcal{M}', x \models \varphi$ iff $\varphi \in x$. Take the case of $\Box \varphi$ then. If $\Box \varphi \in x$, then $\mathcal{M}_\mathbf{K}, x \models \Box \varphi$, and since $R' \subseteq R_\mathbf{K}$, that means $\mathcal{M}', x \models \Box \varphi$. If now $\Box \varphi \not\in x$, then $\mathcal{M}_\mathbf{K}, x \not\models \Box \varphi$, which means for some $y \in R_\mathbf{K}(x)$, $\mathcal{M}_\mathbf{K}, y \not\models \varphi$. Suppose $y \not\in W'$. Then, for some $\psi$, $\psi \in y$ and $\vdash^\mathbf{L} \neg \psi$. Since $y \in R_\mathbf{K}(x)$, that means $L(x) \subseteq y$. Since $x$ is maximal $\mathbf{L}$-consistent, $\neg \psi \in x$, and since $\vdash^\mathbf{L} \neg \psi$, by \emph{Necessitation} we get $\vdash^\mathbf{L} \Box \neg \psi$. Thus, $\Box \neg \psi \in x$, which means $\neg \psi \in y$, contradicting $y$'s consistency. Therefore, $y \in W'$, and by induction hypothesis, $\mathcal{M}', y \not\models \varphi$, so that $\mathcal{M}', x \not\models \Box \varphi$.
\end{proof}

The construction of the canonical model for $\mathbf{L}$ differs from the usual, but it nevertheless yields the same resulting model. The relation is intensionally defined in the same way for the canonical model of $\mathbf{K}$ and that of, in the usual construction, $\mathbf{L}$. Furthermore, since $\mathbf{K} \subseteq \mathbf{L}$, if any set of formulas is maximal $\mathbf{L}$-consistent, then it is maximal consistent, and therefore $W ' = W_\mathbf{L} \subseteq W_\mathbf{K}$, as defined in Theorem \ref{theorem canonical model}. It follows easily that $R_\mathbf{L} = R_\mathbf{K} \restriction_{W'} = R'$. We shall now see how to do something similar, using the canonical model for $\mathbf{K}$ to build a model for a logic $\mathbf{L}^\lhd$ of $\mathcal{L}^\lhd$.

Let $\mathcal{M}_\mathbf{K} = \langle W_\mathbf{K}, R_\mathbf{K}, V_\mathbf{K} \rangle$ be the canonical model of $\mathbf{K}$. We may construct a model for $\mathbf{L}^\lhd$ from it such that, if it makes $\varphi$ true, then $\vdash^{\mathbf{L}^\lhd} \varphi$. Let $\mathbf{L}_t$ be the smallest logic containing $t[\mathbf{L}^\lhd] = \{t(\varphi) \mid\ \vdash^{\mathbf{L}^\lhd} \varphi\}$ -- that is, it is the closure of $t[\mathbf{L}^\lhd]$ under \emph{Modus Ponens}, \emph{Necessitation} and Uniform Substitution. Define $\mathcal{M}_{\mathbf{L}_t}$ as in Theorem \ref{theorem canonical model}, that is, $\mathcal{M}_{\mathbf{L}_t} = \langle W_{\mathbf{L}_t}, R_{\mathbf{L}_t}, V_\mathbf{K} \rangle$, where $W_{\mathbf{L}_t} = \{x \mid  x\ \text{is maximal}\ \mathbf{L}_t \text{-consistent}\}$ and $R_{\mathbf{L}_t} = R_\mathbf{K} \restriction_{W_{\mathbf{L}_t}}$.

\begin{definition}[Pseudonormal modal logic]
Let $\mathbf{L}^\lhd$ be a modal logic of $\mathcal{L}^\lhd$ and $t : \mathcal{L}^\lhd \to \mathcal{L}^\Box$. We say $\mathbf{L}^\lhd$ is \emph{pseudonormal w.r.t. $\mathbf{L}_t$} if for any $\varphi \in Form_{\mathcal{L}^\lhd}$, if $t(\varphi) = \Box \psi$ for some $\psi \in Form_{\mathcal{L}^\Box}$, then whenever $\vdash^{\mathbf{L}_t} \psi$, if there is $\beta \in Form_{\mathcal{L}^\lhd}$ such that $\vdash^{\mathbf{L}_t} t(\beta) \leftrightarrow \psi$, then $\vdash^{\mathbf{L}^\lhd} \beta$ implies $\vdash^{\mathbf{L}^\lhd} \varphi$. 
\end{definition}

In other words, a logic is pseudonormal w.r.t. $\mathbf{L}_t$ if it has some form of the rule of \emph{Necessitation} whenever formulas whose outermost operator is $\Box$ is in the image of $t$ -- that is, if $t$ reflects, in some sense, the rule from its image to its domain. Notice, however, that this requirement is only necessary when pseudonormality is not vacuously satisfied.

\begin{proposition}\label{proposition surjective theorem reflect}
Let $t : \mathcal{L}^\lhd \to \mathcal{L}^\Box$ be compositional and surjective up to equivalence w.r.t. $\mathbf{L}^\lhd$, and $\mathbf{L}^\lhd$ be a logic of $\mathcal{L}^\lhd$ pseudonormal w.r.t. $\mathbf{L}_t$. Then $t$ is theorem reflecting from $\mathbf{L}^\lhd$ to $\mathbf{L}_t$.
\end{proposition}

\begin{proof}
Let $\varphi \in Form_{\mathcal{L}^\lhd}$. We show that if $\vdash^{\mathbf{L}_t} t(\varphi)$, then we must $\vdash^{\mathbf{L}^\lhd} \varphi$. If $t(\varphi) \in t[\mathbf{L}^\lhd]$, then the conclusion is trivial. Otherwise, we proceed by induction on the length of the proof of $t(\varphi)$. The base case is trivial. For the induction step, let now $\vdash^{\mathbf{L}_t} t(\varphi)$, and let the length of the proof be $n+1$. Then either $\vdash^{\mathbf{L}_t} \psi \rightarrow t(\varphi)$, $\vdash^{\mathbf{L}_t} \psi$, and the last step is an application of \emph{Modus Ponens}, or $t(\varphi) = \Box \gamma$ for some $\gamma \in Form_{\mathcal{L}^\Box}$, $\vdash^{\mathbf{L}_t} \gamma$, and the last step is an application of \emph{Necessitation}. If the former, since $t$ is surjective up to equivalence w.r.t. $\mathbf{L}^\lhd$ there is $\alpha \in Form_{\mathcal{L}^\lhd}$ such that $\vdash^{\mathbf{L}_t} \psi \leftrightarrow t(\alpha)$, and thus $\vdash^{\mathbf{L}_t} t(\alpha) \rightarrow t(\varphi)$ and $\vdash^{\mathbf{L}_t} t(\alpha)$. Since $t$ is compositional, $t(\alpha) \rightarrow t(\varphi) = t(\alpha \rightarrow \varphi)$, so by induction hypothesis, $\vdash^{\mathbf{L}^\lhd} \alpha \rightarrow \varphi$ and $\vdash^{\mathbf{L}^\lhd} \alpha$, so by \emph{Modus Ponens}, $\vdash^{\mathbf{L}^\lhd} \varphi$. If the latter, since $h$ is surjective up to equivalence there is $\beta \in Form_{\mathcal{L}^\lhd}$ such that $\vdash^{\mathbf{L}_t} t(\beta) \leftrightarrow \gamma$, so that $\vdash^{\mathbf{L}_t} t(\beta)$, and by induction hypothesis, $\vdash^{\mathbf{L}^\lhd} \beta$. Since $\mathbf{L}^\lhd$ is pseudonormal w.r.t. $\mathbf{L}_t$, that means $\vdash^{\mathbf{L}^\lhd} \varphi$.
\end{proof}

Notice, in general, the fact whether $\mathbf{L}^\lhd$ is pseudonormal w.r.t. $\mathbf{L}_t$ relies on both the specific axioms and inference rules of $\mathbf{L}^\lhd$, and on the definition of $t$. Many times, $\mathbf{L}^\lhd$ may be vacuously pseudonormal w.r.t. $\mathbf{L}_t$, that is, $t$ may not have a formula of the form $\Box \varphi$ in its image. For example, all of the operators previously considered -- $\circ$, $\Delta$, $W$ and $U$ -- are definable in $\mathcal{L}^\Box$ by schema whose main operator is either conjunction or disjunction, so that the natural truth-preserving translations from each of their respective languages to $\mathcal{L}^\Box$ has no formula whose main operator is $\Box$ in their image, which means the logics of their languages are, in general, vacuously pseudonormal.

Another fact that may be pointed out is the following:

\begin{proposition}\label{proposition surjectivity}
Let $\mathbf{L}^\lhd$ be a logic of $\mathcal{L}^\lhd$ and $t : \mathcal{L}^\lhd \to \mathcal{L}^\Box$ be compositional. If for any $\varphi \in Form_{\mathcal{L}^\lhd}$, $\vdash^{\mathbf{L}_t} t(\lhd \varphi) \leftrightarrow \Box t(\varphi)$, then $t$ is surjective up to equivalence w.r.t. $\mathbf{L}^\lhd$.
\end{proposition}

\begin{proof}
By induction on the complexity of $\varphi \in Form_{\mathcal{L}^\Box}$, we show there is $\psi \in Form_{\mathcal{L}^\lhd}$ such that $\vdash^{\mathbf{L}_t} \varphi \leftrightarrow t(\psi)$. By compositionality, the non-modal cases are straightforward. Let $\varphi = \Box \gamma$ for some $\gamma \in Form_{\mathcal{L}^\Box}$. By induction hypothesis there is $\alpha \in Form_{\mathcal{L}^\lhd}$ such that $\vdash^{\mathbf{L}_t} \gamma \leftrightarrow t(\alpha)$, so by \emph{Necessitation}, $\vdash^{\mathbf{L}_t} \Box \gamma \leftrightarrow \Box t(\alpha)$. But by assumption, $\vdash^{\mathbf{L}_t} t(\lhd \alpha) \leftrightarrow \Box t(\alpha)$, and thus $\vdash^{\mathbf{L}_t} \Box \gamma \leftrightarrow t(\lhd \alpha)$.
\end{proof}

\begin{proposition}\label{proposition transitive theorem reflect}
Let $\mathbf{L} \subseteq \mathbf{L}'$ be normal modal logics, and $\mathbf{L}^\lhd$ be a logic of $\mathcal{L}^\lhd$. If $t : \mathcal{L}^\lhd \to \mathcal{L}^\Box$ is theorem reflecting from $\mathbf{L}^\lhd$ to $\mathbf{L}'$, then $t$ is also theorem reflecting from $\mathbf{L}^\lhd$ to $\mathbf{L}$.
\end{proposition}

\begin{proof}
Since $\mathbf{L} \subseteq \mathbf{L}'$, $\vdash^\mathbf{L} \varphi$ implies $\vdash^{\mathbf{L}'} \varphi$. Let $\psi \in Form_{\mathcal{L}^\lhd}$. Since $t$ is theorem reflecting from $\mathbf{L}^\lhd$ to $\mathbf{L}'$, $\vdash^{\mathbf{L}'} t(\psi)$ implies $\vdash^{\mathbf{L}^\lhd} \psi$. Thus, we get that $\vdash^\mathbf{L} t(\psi)$ implies $\vdash^{\mathbf{L}^\lhd} \psi$.
\end{proof}

Proposition \ref{proposition surjective theorem reflect} offers a path to checking whether a translation $t : \mathcal{L}^\lhd \to \mathcal{L}^\Box$ is theorem reflecting, and Propositions \ref{proposition surjectivity} and \ref{proposition transitive theorem reflect} may offer a way to show the necessary conditions for that.

\begin{lemma}\label{lemma model for l}
Let $t$ be theorem reflecting from $\mathbf{L}^\lhd$ to $\mathbf{L}_t$ and truth-preserving in $\mathbb{C}_{\mathbf{L}_t}$, and let $\mathbf{L}_t$ be canonical. For any $\varphi \in Form_{\mathcal{L}^\lhd}$, if $\mathcal{M}_{\mathbf{L}_t} \models \varphi$, then $\vdash^{\mathbf{L}^\lhd} \varphi$.
\end{lemma}

\begin{proof}
Suppose $\not\vdash^{\mathbf{L}^\lhd} \varphi$. Since $t$ is theorem reflecting, $\not\vdash^{\mathbf{L}_t} t(\varphi)$, so $\{\neg t(\varphi)\}$ is $\mathbf{L}_t$-consistent, and therefore $\neg t(\varphi) \in x$ for some $x \in W_{\mathbf{L}_t}$. By Theorem \ref{theorem canonical model}, that means $\mathcal{M}_{\mathbf{L}_t}, x \not\models t(\varphi)$. Since $t$ is truth preserving in $\mathbb{C}_{\mathbf{L}_t}$ and $\mathbf{L}_t$ is canonical, that means $\mathcal{M}_{\mathbf{L}_t}, x \not\models \varphi$. Therefore, $\mathcal{M}_{\mathbf{L}_t} \not\models \varphi$.
\end{proof}

\begin{theorem}\label{theorem completeness 2}
Let $t$ be truth-preserving in $\mathbb{C}_{\mathbf{L}_t}$ and theorem reflecting from $\mathbf{L}^\lhd$ to $\mathbf{L}_t$. Let also $\mathbf{L}_t$ be canonical. Then $\mathbf{L}^\lhd$ is complete w.r.t. $\mathbb{C}_{\mathbf{L}_t}$.
\end{theorem}

\begin{proof}
Let $\varphi \in Form_{\mathcal{L}^\lhd}$. If $\mathbb{C}_{\mathbf{L}_t} \models \varphi$, then, since $\mathbf{L}_t$ is canonical, $\mathcal{M}_{\mathbf{L}_t} \models \varphi$, so by Lemma \ref{lemma model for l}, $\vdash^{\mathbf{L}^\lhd} \varphi$.
\end{proof}

Therefore, we have \emph{Completeness strategy $2$} of $\mathbf{L}^\lhd$ w.r.t. $\mathbb{C}_{\mathbf{L}_t}$ by:

\begin{enumerate}
\item defining a mapping $t: \mathcal{L}^\lhd \to \mathcal{L}^\Box$ which is truth-preserving in $\mathbb{C}_{\mathbf{L}_t}$ and theorem reflecting from $\mathbf{L}^\lhd$ to $\mathbf{L}_t$;
\item showing $\mathbf{L}_t$ is canonical.
\end{enumerate}

The task of finding such a theorem reflecting translation $t$ may seem difficult, but in light of Propositions \ref{proposition surjective theorem reflect} and \ref{proposition surjectivity} the difficulty may be partially averted.

Notice the stronger requirement of canonicity of $\mathbf{L}_t$ may be dropped in favour of the simple completeness of it w.r.t. $\mathbb{C}_{\mathbf{L}_t}$:

\begin{theorem}\label{theorem completeness 4}
Let $t$ be truth-preserving in $\mathbb{C}_{\mathbf{L}_t}$ and theorem reflecting from $\mathbf{L}^\lhd$ to $\mathbf{L}_t$. Let also $\mathbf{L}_t$ be complete w.r.t. $\mathbb{C}_{\mathbf{L}_t}$. Then $\mathbf{L}^\lhd$ is complete w.r.t. $\mathbb{C}_{\mathbf{L}_t}$.
\end{theorem}

\begin{proof}
Let $\varphi \in Form_{\mathcal{L}^\lhd}$ and suppose $\mathbb{C}_{\mathbf{L}_t} \models \varphi$. Since $t$ is truth-preserving in that class of frames, $\mathbb{C}_{\mathbf{L}_t} \models t(\varphi)$, and since $\mathbf{L}_t$ is complete with respect to that class, $\vdash^{\mathbf{L}_t} t(\varphi)$. But since $t$ is theorem reflecting from $\mathbf{L}^\lhd$ to $\mathbf{L}_t$, that means $\vdash^{\mathbf{L}^\lhd} \varphi$.
\end{proof}

That modification allows the slightly more genera \emph{Completeness strategy $2.5$} of $\mathbf{L}^\lhd$ w.r.t. $\mathbb{C}_{\mathbf{L}_t}$:

\begin{enumerate}
\item defining a mapping $t: \mathcal{L}^\lhd \to \mathcal{L}^\Box$ which is truth-preserving in $\mathbb{C}_{\mathbf{L}_t}$ and theorem reflecting from $\mathbf{L}^\lhd$ to $\mathbf{L}_t$;
\item showing $\mathbf{L}_t$ is complete w.r.t. $\mathbb{C}_{\mathbf{L}_t}$.
\end{enumerate}

Notice the requirement of $t$ being theorem reflecting may be inverted if the translation goes the other way around, as in the previous method. Let us spell that out. Let $\mathbf{L}$ be a normal modal logic, $h : \mathcal{L}^\Box \to \mathcal{L}^\lhd$, and $\mathbf{L}_h^\lhd$ be as previously defined. Define $\mathcal{M}_\mathbf{L} = \langle W_\mathbf{L}, R_\mathbf{L}, V_\mathbf{K} \rangle$ as in Theorem \ref{theorem canonical model}.

\begin{lemma}\label{lemma completeness 3}
Let $h$ be theorem preserving from $\mathbf{L}$ to $\mathbf{L}^\lhd_h$, truth-preserving in $\mathbf{C}_\mathbf{L}$, and surjective up to equivalence w.r.t. 
$\mathbf{L}$. Let also $\mathbf{L}$ be canonical, and $\mathbf{L}^\lhd_h$ be sound w.r.t. $\mathbb{C}_\mathbf{L}$. Then, for any $\varphi \in Form_{\mathcal{L}^\lhd}$, if $\mathcal{M}_\mathbf{L} \models \varphi$, then $\vdash^{\mathbf{L}^\lhd_h} \varphi$.
\end{lemma}

\begin{proof}
Suppose $\not\vdash^{\mathbf{L}^\lhd_h} \varphi$. Since $h$ is surjective up to equivalence w.r.t. $\mathbf{L}$, $\vdash^{\mathbf{L}^\lhd_h} \varphi \leftrightarrow h(\psi)$ for some $\psi \in Form_{\mathcal{L}^\Box}$, which means $\not\vdash^{\mathbf{L}^\lhd_h} h(\psi)$. Since $h$ is theorem preserving, $\{\neg \psi\}$ must be $\mathbf{L}$-consistent, so that $\neg \psi \in x$ for some $x \in W_\mathbf{L}$. By Theorem \ref{theorem canonical model}, we get $\mathcal{M}_{\mathbf{L}}, x \not\models \psi$; since $h$ is truth preserving in $\mathbb{C}_\mathbf{L}$ and $\mathbf{L}$ is canonical, we get $\mathcal{M}_{\mathbf{L}}, x \not\models h(\psi)$, and thus $\mathcal{M}_\mathbf{L} \not\models h(\psi)$. Since $\mathbf{L}^\lhd_h$ is sound w.r.t. $\mathbb{C}_\mathbf{L}$, $\mathbb{C}_\mathbf{L} \models \varphi \leftrightarrow h(\psi)$, and thus $\mathcal{M}_\mathbf{L} \not\models \varphi$.
\end{proof}

\begin{theorem}\label{theorem completeness 3}
Let $h$ be theorem preserving from $\mathbf{L}$ to $\mathbf{L}^\lhd_h$, truth-preserving in $\mathbf{C}_\mathbf{L}$, and surjective up to equivalence w.r.t. 
$\mathbf{L}$. Let also $\mathbf{L}$ be canonical, and $\mathbf{L}^\lhd_h$ be sound w.r.t. $\mathbb{C}_\mathbf{L}$. Then, $\mathbf{L}^\lhd_h$ is complete w.r.t. $\mathbb{C}_\mathbf{L}$.
\end{theorem}

\begin{proof}
Let $\varphi \in Form_{\mathcal{L}^\lhd}$. If $\mathbb{C}_\mathbf{L} \models \varphi$, since $\mathbf{L}$ is canonical, $\mathcal{M}_\mathbf{L} \models \varphi$, so by Lemma \ref{lemma completeness 3}, $\vdash^{\mathbf{L}^\lhd_h} \varphi$.
\end{proof}

In this way, we have \emph{Completeness strategy $3$} of $\mathbf{L}^\lhd_h$ w.r.t. $\mathbb{C}_\mathbf{L}$ by:

\begin{enumerate}
\item defining a mapping $h: \mathcal{L}^\Box \to \mathcal{L}^\lhd$ which is theorem preserving from $\mathbf{L}$ to $\mathbf{L}^\lhd_h$, surjective up to equivalence w.r.t. $\mathbf{L}$ and truth-preserving in $\mathbb{C}_\mathbf{L}$;
\item showing $\mathbf{L}$ is canonical w.r.t. $\mathbb{C}_\mathbf{L}$;
\item showing $\mathbf{L}^\lhd_h$ is sound w.r.t. $\mathbb{C}_\mathbf{L}$.
\end{enumerate}

Notice this strategy makes less assumptions than the first completeness strategy, for it does require the definition of a canonical model for the logic of $\mathcal{L}^\lhd$, nor $h$ to be boxlike w.r.t. $\mathbf{L}^\lhd_h$. The three last strategies shift the difficulty away from defining a canonical model for a logic of $\mathcal{L}^\lhd$, a non-trivial and -- as the literature shows.

\section{Worked out applications}\label{sec application}

To illustrate the results of the previous sections, we shall offer an example by defining a new modal operator. Let $\mathcal{L}^\delta$ be the language generated by the operator $\delta$, defined in the following way:\footnote{Given a relational frame $\langle W, R \rangle$, $w \in W$ and $n \in \omega$, we define the set $R^n(w)$ as usual, that is, as the set of worlds in $W$ separated from $w$ by $n$ steps.}

\medskip

\begin{tabular}{lll}
$\mathcal{M}, w \models \delta \varphi$ & iff & either $\mathcal{M}, w \not\models \varphi$ or $\forall u \in R^2(w)(\mathcal{M}, u \models \varphi)$\\
 & iff & $\mathcal{M}, w \models \varphi \rightarrow \Box \Box \varphi$
\end{tabular}

\medskip

Consider now the compositional mapping $h : \mathcal{L}^\Box \to \mathcal{L}^\delta$ recursively defined by:\footnote{This mapping is also known as the \emph{boxdot translation} (\cite{pelletier}, \cite{humberstone2005}, \cite{french2009}).}

\medskip

\begin{tabular}{lll}
$h(p)$ & $=$ & $p$\\
$h(\neg \varphi)$ & $=$ & $\neg h(\varphi)$\\
$h(\varphi \wedge \psi)$ & $=$ & $h(\varphi) \wedge h(\psi)$\\
$h(\Box \varphi)$ & $=$ & $h(\varphi) \wedge \delta h(\varphi)$
\end{tabular}

\medskip

\begin{proposition}\label{prop translation}
$h$ is truth-preserving in the class $\mathbb{C}_\mathbf{T4}$ of reflexive transitive frames.
\end{proposition}

\begin{proof}
By induction on the complexity of $\varphi \in Form_{\mathcal{L}^\Box}$, we show $\mathbb{C}_\mathbf{T4} \models \varphi$ iff $\mathbb{C}_\mathbf{T4} \models h(\varphi)$. The non-modal cases are straightforward. Let now $\langle W, R \rangle \in \mathbb{C}_\mathbf{T4}$, $\mathcal{M} = \langle W, R, V \rangle$, and $w \in W$. Suppose $\mathcal{M}, w \models \Box \varphi$. Then, for any $u \in R(w)$, $\mathcal{M}, u \models \varphi$. Since the frame is reflexive, $w \in R(w)$, and so $\mathcal{M}, w \models \varphi$. By induction hypothesis, $\mathcal{M}, w \models h(\varphi)$. Furthermore, since the frame is also transitive, $R(w) = R^2(w)$, so we obtain that for any $u \in R^2(w)$, $\mathcal{M}, u \models \varphi$. By induction hypothesis, for any $u \in R^2(w)$, $\mathcal{M}, u \models h(\varphi)$, and therefore, by definition, $\mathcal{M}, w \models \delta h(\varphi)$. Thus, $\mathcal{M}, w \models h(\varphi) \wedge \delta h(\varphi)$. Let now $\mathcal{M}, w \models h(\Box \varphi)$, that is, $\mathcal{M}, w \models h(\varphi) \wedge \delta h(\varphi)$. Then, for any $u \in R^2(w)$, $\mathcal{M}, u \models h(\varphi)$. But, since $R(w) = R^2(w)$, that means for any $u \in R(w)$, $\mathcal{M}, u \models h(\varphi)$. By induction hypothesis, we obtain that for any $u \in R(w)$, $\mathcal{M}, u \models \varphi$, and thus $\mathcal{M}, w \models \Box \varphi$.
\end{proof}

Let $\mathbf{T4}$ be the extension of $\mathbf{K}$ by axioms $\mathrm{T} = \Box \varphi \rightarrow \varphi$ and $\mathrm{4} = \Box \varphi \rightarrow \Box \Box \varphi$, and $\mathbf{T4}_h$ the smallest logic of $\mathcal{L}^\delta$ containing $h[\mathbf{T4}]$. Then, already, by using the Soundness strategy $1$:

\begin{corollary}\label{corollary soudness 44}
Let $\mathbf{L}^\delta \supseteq \mathbf{T4}_h$. Then $\mathbf{L}^\delta$ is sound w.r.t. $\mathbb{C}_\mathbf{L}$. 
\end{corollary}

\begin{proof}
By Proposition \ref{prop translation} and Theorem \ref{theorem soundness 1}.
\end{proof}

For completeness, let us use the Completeness strategy $1$. Translate the logic $\mathbf{T4}$ into $\mathcal{L}^\delta$. That means its axioms and inference rules become the translated versions of the normal modal logic ones (skipping \emph{Modus Ponens}, the classical propositional logic axioms, and Uniform Substitution):

\begin{itemize}
\item[$\mathrm{K}_h$] $\big{(} (\varphi \rightarrow \psi) \wedge \delta (\varphi \rightarrow \psi) \big{)} \to \big{(} (\varphi \wedge \delta \varphi) \rightarrow (\psi \wedge \delta \psi) \big{)}$
\item[$\mathrm{T}_h$] $(\varphi \wedge \delta \varphi) \rightarrow \varphi$
\item[$\mathrm{4}_h$] $(\varphi \wedge \delta \varphi) \rightarrow \big{(} (\varphi \wedge \delta \varphi) \wedge \delta (\varphi \wedge \delta \varphi) \big{)}$
\item[$Nec_h$] from $\vdash \varphi$ infer $\vdash \varphi \wedge \delta \varphi$
\end{itemize}

Let $\mathbf{L} \supseteq \mathbf{T4}$. Then, notice, by the very structure of the translated axioms, $h$ is theorem-preserving from $\mathbf{L}$ to $\mathbf{L}^\delta_h$, as any proof in $\mathbf{L}$ of $\varphi$ can be transformed straightforwardly into a proof of $h(\varphi)$ in $\mathbf{L}^\delta_h$.

We now need to define the canonical model of $\mathbf{L}^\delta_h$. For $\Gamma \subseteq Form_{\mathcal{L}^\delta}$, let $L^\delta(\Gamma) = \{\varphi \mid \varphi \wedge \delta \varphi \in \Gamma\}$. Let $\mathcal{M}_{\mathbf{L}^\delta_h} = \langle \mathrm{W}_{\mathbf{L}^\delta_h}, \mathrm{R}_{\mathbf{L}^\delta_h}, \mathrm{V}_{\mathbf{L}^\delta_h} \rangle$, where $\mathrm{W}_{\mathbf{L}_h^\delta}$ is the set of all maximal $\mathbf{L}_h^\delta$-consistent sets of formulas, $x \mathrm{R}_{\mathbf{L}_h^\delta} y$ iff $L^\delta(x) \subseteq y$, and $\mathrm{V}_{\mathbf{L}_h^\delta}(p) = \{x \in \mathrm{W}_{\mathbf{L}_h^\delta} \mid p \in x\}$.

\begin{lemma}\label{lemma h is boxlike}
Let $\mathbf{L} \supseteq \mathbf{T4}$. Then $h$ is boxlike w.r.t. $\mathbf{L}_h^\delta$.
\end{lemma}

\begin{proof}
Suppose $x \mathrm{R}_{\mathbf{L}^\lhd_h} y$, let $\varphi \in Form_{\mathcal{L}^\Box}$, and $h(\Box \varphi) \in x$ (that is, $h(\varphi) \wedge \delta h(\varphi) \in x$). By definition, that means $h(\varphi) \in L^\delta(x) \subseteq y$. For the other way around, if for any $\varphi \in Form_{\mathcal{L}^\Box}$, $h(\Box \varphi) \in x$ (that is, $h(\varphi) \wedge \delta h(\varphi) \in x$) implies $h(\varphi) \in y$, that means $L^\delta(x) \subseteq y$, and so $x \mathrm{R}_{\mathbf{L}^\lhd_h} y$.
\end{proof}

\begin{lemma}\label{lemma h is surjective}
Let $\mathbf{L} \supseteq \mathbf{T4}$. Then $h$ is surjective up to equivalence w.r.t. $\mathbf{L}$.
\end{lemma}

\begin{proof}
By induction on the complexity of $\varphi \in Form_{\mathcal{L}^\delta}$, we show there is $\psi \in Form_{\mathcal{L}^\Box}$ such that $\vdash^{\mathbf{L}_h^\delta} \varphi \leftrightarrow h(\psi)$. By compositionality, the non-modal cases are straightforward, so suppose $\varphi = \delta \gamma$. By induction hypothesis, there is $\alpha \in Form_{\mathcal{L}^\Box}$ such that $\vdash^{\mathbf{L}^\delta_h} \gamma \leftrightarrow h(\alpha)$. Since $\mathbf{L} \supseteq \mathbf{T4}$, we have $\mathbf{L}^\lhd_h \supseteq \mathbf{T4}_h$, so by $Nec_h$, $\vdash^{\mathbf{L}^\delta_h} (\gamma \leftrightarrow h(\alpha)) \wedge \delta (\gamma \leftrightarrow h(\alpha))$. By \emph{Modus Ponens} and $\mathrm{K}_h$, $\vdash^{\mathbf{L}^\delta_h} \delta \gamma \leftrightarrow \delta h(\alpha)$.
\end{proof}

We, thus, may easily take away the following result:

\begin{corollary}
Let $\mathbf{L} \supseteq \mathbf{T4}$, and $\mathbf{L}$ be canonical and complete w.r.t. $\mathbb{C}_\mathbf{L}$. Then, $\mathbf{L}^\delta_h$ is complete w.r.t. $\mathbb{C}_\mathbf{L}$.
\end{corollary}

\begin{proof}
By Proposition \ref{prop translation}, $h$ preserves truth in $\mathbb{C}_\mathbf{L}$. As we have seen, it is also theorem preserving and compositional, and by Lemma \ref{lemma h is boxlike}, boxlike w.r.t. $\mathbf{L}^\delta_h$, and by Lemma \ref{lemma h is surjective}, $h$ is surjective up to equivalence w.r.t. $\mathbf{L}$. Since $\mathbf{L}$ is canonical and complete w.r.t. $\mathbb{C}_\mathbf{L}$, by Corollary \ref{corollary soudness 44}, $\mathbf{L}^\delta_h$ is sound w.r.t. $\mathbb{C}_\mathbf{L}$, so by Theorem \ref{theorem completeness 1}, $\mathbf{L}_h^\delta$ must be complete w.r.t. $\mathbb{C}_\mathbf{L}$.
\end{proof}

We may thus see economy of this strategy: given the work on Sections \ref{sec soundness} and \ref{sec completeness}, in less than two pages, soundness and completeness may be given by a newly introduced logic.

Let us now showcase other completeness strategies. We shall recover the soundness and completeness of the minimal logic of $\mathcal{L}^\circ$, which we may call $\mathbf{K}^\circ$. The logic is defined by the following axioms and inference rule:\footnote{We here reproduce the axiomatization offered in \cite{steinsvold2008a} and \cite{gilbertventuri2016}, rather than the one originally provided in \cite{marcos2005}, for the former is more economical.}

\begin{itemize}
\item[$b0$] $\circ \top$
\item[$b1$] $\neg \circ \varphi \rightarrow \varphi$
\item[$b2$] $(\circ \varphi \wedge \circ \psi) \rightarrow \circ(\varphi \wedge \psi)$
\item[$bN$] from $\vdash^{\mathbf{K}^\circ} \varphi \rightarrow \psi$ infer $\vdash^{\mathbf{K}^\circ} (\circ \varphi \wedge \varphi) \rightarrow (\circ \psi \wedge \psi)$
\end{itemize}

\noindent
We further recall the follow derived rule $\mathbf{K}^\circ$ \cite{steinsvold2008a}:

\begin{itemize}
\item[$bNec$] from $\vdash^{\mathbf{K}^\circ} \varphi$, one may obtain $\vdash^{\mathbf{K}^\circ} \circ \varphi$
\end{itemize}

Now, to follow the completeness strategies, we need a translation $t : \mathcal{L}^\circ \to \mathcal{L}^\Box$ which is theorem reflecting. To prove $t$ reflects theoremhood, it helps to show $t$ is surjective up to equivalence -- a task which is easier when a translation has formulas whose main operator is $\Box$ in its image --, and that $t$ is compositional with respect to the conditional -- that is, $t(\varphi \rightarrow \psi) = t(\varphi) \rightarrow t(\psi)$. There are two non-trivial ways we may do that.

The first is to let $\mathcal{L}^\circ$ have each classical operator as a primitive of the language, and let the equivalences between them be imported as axioms from classical propositional logic. Consider then the translation given by:

\medskip

\begin{tabular}{lll}
$t(p)$ & $=$ & $p$\\
$t(\neg \varphi)$ & $=$ & $\neg t(\varphi)$\\
$t(\varphi \vee \psi)$ & $=$ & $t(\varphi) \vee t(\psi)$\\
$t(\varphi \rightarrow \psi)$ & $=$ & $t(\varphi) \rightarrow t(\psi)$\\
$t(\varphi \leftrightarrow \psi)$ & $=$ & $t(\varphi) \leftrightarrow t(\psi)$\\
$t(\varphi \wedge \psi)$ & $=$ & $\begin{cases}
t(\varphi) \wedge t(\psi),\ \text{if}\ \varphi \neq \circ \psi\ \text{and}\ \psi \neq \circ \varphi\\
\Box t(\varphi),\ \text{if}\ \psi = \circ \varphi\\
\Box t(\psi),\ \text{if}\ \varphi = \circ \psi\\
\end{cases}$\\
$t(\circ \varphi)$ & $=$ & $\neg t(\varphi) \vee \Box t(\varphi)$
\end{tabular}

\medskip

\noindent
In that case $t$, is compositional for all classical operators excluding conjunction. Still, one may check compositionality holds for negation and disjunction. The second option is to let negation and conjunction be the only primitive classical operators, and instead define an even less compositional translation:

\medskip

\begin{tabular}{lll}
$t(p)$ & $=$ & $p$\\
$t(\neg \varphi)$ & $=$ & $\begin{cases}
\neg t(\varphi),\ \text{if}\ \varphi \neq \alpha \wedge \neg \beta\\
\neg (t(\alpha) \wedge \neg t(\beta)),\ \text{if}\ \varphi = \alpha \wedge \neg \beta
\end{cases}$\\
$t(\varphi \wedge \psi)$ & $=$ & $\begin{cases}
t(\varphi) \wedge t(\psi),\ \text{if}\ \varphi \neq \circ \psi\ \text{and}\ \psi \neq \circ \varphi\\
\Box t(\varphi),\ \text{if}\ \psi = \circ \varphi\\
\Box t(\psi),\ \text{if}\ \varphi = \circ \psi\\
\end{cases}$\\
$t(\circ \varphi)$ & $=$ & $\neg t(\varphi) \vee \Box t(\varphi)$
\end{tabular}

\medskip

\noindent
In that case, we may see two formulas which are classically equivalent, but which differ in the order of the scope of its nested operators will have different images under $t$. However, we still obtain compositionality for the conditional: since $\varphi \rightarrow \psi := \neg (\varphi \wedge \neg \psi)$, $t(\varphi \rightarrow \psi) = t(\neg (\varphi \wedge \neg \psi)) = \neg (t(\varphi) \wedge \neg t(\psi)) = t(\varphi) \rightarrow t(\psi)$. With either choice of translation, we may follow through:

\begin{proposition}\label{proposition t truth preserving ct}
$t$ is truth-preserving in the class of reflexive frames $\mathbb{C}_\mathbf{T}$.
\end{proposition}

\begin{proof}
By induction on the complexity of $\varphi \in Form_{\mathcal{L}^\circ}$. The only non obvious case is when $\varphi = \psi \wedge \circ \psi$. However, $\mathbb{C}_\mathbf{T} \models \Box t(\psi) \rightarrow t(\psi)$, in which case we may obtain both directions of the satisfaction of $\varphi$.
\end{proof}

\begin{lemma}\label{lemma circ t is surjective}
$t$ is surjective up to equivalence w.r.t. $\mathbf{T}$.
\end{lemma}

\begin{proof}
By induction on the complexity of $\varphi \in Form_{\mathcal{L}^\Box}$ we show there is $\psi \in Form_{\mathcal{L}^\circ}$ such that $\vdash^{\mathbf{T}} \varphi \leftrightarrow t(\psi)$. The atomic case and the case of negation are straightforward. If $\varphi = \psi \wedge \gamma$, by induction hypothesis there are $\alpha, \beta \in Form_{\mathcal{L}^\circ}$ such that $\vdash^\mathbf{T} \psi \leftrightarrow t(\alpha)$ and $\vdash^\mathbf{T} \gamma \leftrightarrow t(\beta)$, so that $\vdash^\mathbf{T} (\psi \wedge \gamma) \leftrightarrow (t(\alpha) \wedge t(\beta))$. If $\alpha \neq \circ \beta$ and $\beta \neq \circ \alpha$, then by the definition of $t$, $t(\alpha) \wedge t(\beta) = t(\alpha \wedge \beta)$, and thus $\vdash^\mathbf{T} (\psi \wedge \gamma) \leftrightarrow t(\alpha \wedge \beta)$. Otherwise, let $\alpha = \circ \beta$. Then $t(\alpha) \wedge t(\beta) = t(\circ \beta) \wedge t(\beta) = (\neg t(\beta) \vee \Box t(\beta)) \wedge t(\beta)$. However, $\vdash^\mathbf{T} ((\neg t(\beta) \vee \Box t(\beta)) \wedge t(\beta)) \leftrightarrow (\Box t(\beta) \wedge t(\beta))$, and $\Box t(\beta) \wedge t(\beta) = t(\circ \beta \wedge \beta) \wedge t(\beta) = t((\circ \beta \wedge \beta) \wedge \beta)$. Therefore, $\vdash^\mathbf{T} (\psi \wedge \gamma) \leftrightarrow t((\circ \beta \wedge \beta) \wedge \beta)$. In case $\beta = \circ \alpha$, we have an analogous conclusion. If now $\varphi = \Box \gamma$, by induction hypothesis there is $\alpha \in Form_{\mathcal{L}^\circ}$ such that $\vdash^{\mathbf{T}} \gamma \leftrightarrow t(\alpha)$, and thus $\vdash^{\mathbf{T}} \Box \gamma \leftrightarrow \Box t(\alpha)$. Since $\Box t(\alpha) = t(\circ \alpha \wedge \alpha)$, we finally get $\vdash^\mathbf{T} \Box \gamma \leftrightarrow t(\circ \alpha \wedge \alpha)$.
\end{proof}

\begin{lemma}\label{lemma kt is t}
$\mathbf{T} = \mathbf{K}_t$.
\end{lemma}

\begin{proof}
For one direction, notice the translation of axioms $b0$, $b1$ and $b2$, and of rule $bN$, are true in $\mathbb{C}_\mathbf{K}$, and thus derivable in $\mathbf{T}$.\footnote{See \cite{gilbertventuri2016}.} For the other, notice that axiom $\mathrm{T}$ is in $\mathbf{K}_t$: $(\varphi \wedge \circ \varphi) \rightarrow \varphi$ is a tautology, and $t((\varphi \wedge \circ \varphi) \rightarrow \varphi) = t(\varphi \wedge \circ \varphi) \rightarrow t(\varphi) = \Box t(\varphi) \rightarrow t(\varphi)$. Furthermore, axiom $\mathrm{K}$ is in $\mathbf{K}_t$: $\vdash^{\mathbf{K}^\circ} ((\varphi \rightarrow \psi) \wedge \circ (\varphi \rightarrow \psi)) \rightarrow ((\varphi \wedge \circ \varphi) \rightarrow (\psi \wedge \circ \psi))$, and $t(((\varphi \rightarrow \psi) \wedge \circ (\varphi \rightarrow \psi)) \rightarrow ((\varphi \wedge \circ \varphi) \rightarrow (\psi \wedge \circ \psi))) = t((\varphi \rightarrow \psi) \wedge \circ (\varphi \rightarrow \psi)) \rightarrow t((\varphi \wedge \circ \varphi) \rightarrow (\psi \wedge \circ \psi)) = \Box (t(\varphi) \rightarrow t(\psi)) \rightarrow (\Box t(\varphi) \rightarrow \Box t(\psi))$.\footnote{The proof of this sentence in $\mathbf{K}^\circ$ is easily obtainable by applying $bN$ to the second conjunct of the antecedent.} Furthermore, by definition, the inference rules of $\mathbf{T}$ (\emph{Modus Ponens}, \emph{Necessitation}, Uniform Substitution) hold in $\mathbf{K}_t$. Thus, $\mathbf{T} \subseteq \mathbf{K}_t$. 
\end{proof}

\begin{corollary}\label{corollary t preserv circ}
$t$ is theorem preserving from $\mathbf{K}^\circ$ to $\mathbf{T}$.
\end{corollary}

\begin{proof}
Straightforward by Lemma \ref{lemma kt is t}.
\end{proof}

We thus may follow the Soundness strategy $2$:

\begin{theorem}\label{theorem soundness k circ}
$\mathbf{K}^\circ$ is sound w.r.t. $\mathbb{C}_\mathbf{T}$.
\end{theorem}

\begin{proof}
We know $\mathcal{L}^\circ \preceq \mathcal{L}^\Box$, and by Lemma \ref{lemma kt is t}, $\mathbf{T} = \mathbf{K}_t$. Furthermore, by Proposition \ref{proposition t truth preserving ct}, $t$ is truth-preserving in $\mathbb{C}_\mathbf{T}$, and by Corollary \ref{corollary t preserv circ}, it is also theorem preserving from $\mathbf{K}^\circ$ to $\mathbf{T}$. Thus, by Theorem \ref{theorem soundness 2}, $\mathbf{K}^\circ$ must be sound w.r.t. $\mathbb{C}_\mathbf{T}$.
\end{proof}

Now, since $t$ is not compositional, we cannot use Proposition \ref{proposition surjective theorem reflect}. Instead, we prove directly, with the above results, that $t$ is theorem-reflecting.

\begin{lemma}\label{lemma t reflect circ}
$t$ is theorem reflecting from $\mathbf{K}^\circ$ to $\mathbf{T}$.
\end{lemma}

\begin{proof}
By induction on the length of the proof of $t(\varphi)$ for $\varphi \in Form_{\mathcal{L}^\circ}$ we show that if $\vdash^\mathbf{T} t(\varphi)$, then we must have $\vdash^{\mathbf{K}^\circ} \varphi$. If the length is $0$, by Lemma \ref{lemma kt is t} the conclusion is trivial. For the inductive step, let $\vdash^\mathbf{T} t(\varphi)$, and let the length of the proof be $n+1$. Then either $\vdash^\mathbf{T} \psi \rightarrow t(\varphi)$, $\vdash^\mathbf{T} \psi$, and the last step is an application of \emph{Modus Ponens}, or $t(\varphi) = \Box \gamma$ for some $\gamma \in Form_{\mathcal{L}^\Box}$, $\vdash^\mathbf{T} \gamma$, and the last step is an application of \emph{Necessitation}. In case of \emph{Modus Ponens}, by Lemma \ref{lemma circ t is surjective} $t$ is surjective up to equivalence w.r.t. $\mathbf{T}$, so there is $\alpha \in Form_{\mathcal{L}^\circ}$ such that $\vdash^\mathbf{T} \psi \leftrightarrow t(\alpha)$, and thus $\vdash^\mathbf{T} t(\alpha) \rightarrow t(\varphi)$ and $\vdash^\mathbf{T} t(\alpha)$. By the definition of $t$ (that is, the compositionality for the conditional), $t(\alpha) \rightarrow t(\varphi) = t(\alpha \rightarrow \varphi)$, so by induction hypothesis, $\vdash^{\mathbf{K}^\circ} \alpha \rightarrow \varphi$ and $\vdash^{\mathbf{K}^\circ} \alpha$, so by \emph{Modus Ponens}, $\vdash^\mathbf{T} \varphi$. In case of \emph{Necessitation}, by the definition of $t$, $\varphi = (\beta \wedge \circ \beta)$ and $\gamma = t(\beta)$ for some $\beta \in Form_{\mathcal{L}^\circ}$. By induction hypothesis, $\vdash^{\mathbf{K}^\circ} \beta$. By $bNec$, $\vdash^{\mathbf{K}^\circ} \circ \beta$, which gives us $\vdash^{\mathbf{K}^\circ} \beta \wedge \circ \beta$.
\end{proof}

At this point, we may follow either Completeness strategy $2$ or $2.5$.

\begin{theorem}\label{theorem completeness circ}
$\mathbf{K}^\circ$ is complete w.r.t. $\mathbb{C}_\mathbf{T}$.
\end{theorem}

\begin{proof}
By Lemma \ref{lemma kt is t}, $\mathbf{K}_t = \mathbf{T}$. $\mathbf{T}$ is canonical, and $t$ is truth-preserving in $\mathbb{C}_\mathbf{T}$. By Lemma \ref{lemma t reflect circ}, $t$ is theorem reflecting from $\mathbf{K}^\circ$ to $\mathbf{T}$, so by Theorem \ref{theorem completeness 2}, we have the conclusion. Alternatively, since $t$ is truth-preserving in $\mathbb{C}_\mathbf{T}$ and theorem reflecting from $\mathbf{K}^\circ$ to $\mathbf{T}$; and since $\mathbf{T}$ is complete w.r.t. $\mathbb{C}_\mathbf{T}$, by Theorem \ref{theorem completeness 3}, we also have the conclusion.
\end{proof}

Thus, we have offered a proof of the completeness of $\mathbf{K}^\circ$ which, unlike those offered in \cite{marcos2005}, \cite{steinsvold2008a} or \cite{gilbertventuri2016}, does not require an explicit definition of a canonical model for $\mathbf{K}^\circ$. Notice the asymmetry between Theorems \ref{theorem soundness k circ} and \ref{theorem completeness circ}, which highlight the collapse of validity between $\mathbb{C}_\mathbf{K}$ and $\mathbb{C}_\mathbf{T}$ in $\mathcal{L}^\circ$.

Let us now try that with a different operator. Consider the operator $U$ of unknowable necessary truths, presented in \cite{yagoventuri2026}, and defined as:

\medskip

\begin{tabular}{lll}
$\mathcal{M}, w \models U \varphi$ & iff & $\mathcal{M}, w \models \varphi$ and $\forall u \in R(w) \exists v \in R(u)(\mathcal{M}, v \not\models \varphi)$\\
 & iff & $\mathcal{M}, w \models \varphi \wedge \Box \Diamond \neg \varphi$
\end{tabular}

\medskip

Consider now the mapping $h : \mathcal{L}^\Box \to \mathcal{L}^U$ defined as:

\medskip

\begin{tabular}{lll}
$h(p)$ & $=$ & $p$\\
$h(\neg \varphi)$ & $=$ & $\neg h(\varphi)$\\
$h(\varphi \wedge \psi)$ & $=$ & $h(\varphi) \wedge h(\psi)$\\
$h(\Box \varphi)$ & $=$ & $h(\varphi) \wedge \neg U h(\varphi)$
\end{tabular}

\medskip

\begin{proposition}\label{prop translation 2}
$h$ is truth-preserving in the class $\mathbb{C}_\mathbf{N}$ of strictly narcissistic frames (those such that for any $w \in W$, $R(w) = \{w\}$).
\end{proposition}

\begin{proof}
By induction on the complexity of $\varphi \in Form_{\mathcal{L}^\Box}$, we show $\mathbb{C}_\mathbf{N} \models \varphi$ iff $\mathbb{C}_\mathbf{N} \models h(\varphi)$. We skip the non-modal cases. Let $\langle W, R \rangle \in \mathbb{C}_\mathbf{N}$, $\mathcal{M} = \langle W, R, V \rangle$, and $w \in W$. Suppose $\mathcal{M}, w \models \Box \varphi$. Then, for any $u \in R(w)$, $\mathcal{M}, u \models \varphi$. By the property of $\mathbb{C}_\mathbb{N}$, $w \in R(w)$, and so $\mathcal{M}, w \models \varphi$. By induction hypothesis, $\mathcal{M}, w \models h(\varphi)$. Since strict narcissism implies reflexivity, $R(w) = R^2(w) = \{w\}$, so we get that $\mathcal{M}, w \models \neg U h(\varphi)$. Thus, $\mathcal{M}, w \models h(\varphi) \wedge \neg U h(\varphi)$. Let now $\mathcal{M}, w \models h(\varphi) \wedge \neg U h(\varphi)$. That means $\mathcal{M}, w \models h(\varphi)$. By induction hypothesis, $\mathcal{M}, w \models \varphi$. Once again, since $R(w) = R^2(w) = \{w\}$, that means $\mathcal{M}, w \models \Box \varphi$.
\end{proof}

Let $\mathbf{N}$ be the extension of $\mathbf{K}$ by axioms $\mathrm{N} = \varphi \rightarrow \Box \varphi$ and $\mathrm{D} = \Diamond \top$, and $\mathbf{N}_h$ the smallest logic of $\mathcal{L}^U$ containing $h[\mathbf{N}]$.\footnote{We note axiom $\mathrm{D}$ usually takes the form $\Box \varphi \rightarrow \Diamond \varphi$, but for simplicity, we use the formulation as given.} Then, by the Soundness strategy $1$:

\begin{corollary}\label{corollary soudness 55}
Let $\mathbf{L}^U \supseteq \mathbf{N}_h$. Then $\mathbf{L}^U$ is sound w.r.t. $\mathbb{C}_{\mathbf{N}}$. 
\end{corollary}

\begin{proof}
By Proposition \ref{prop translation 2} and Theorem \ref{theorem soundness 1}.
\end{proof}

This time, let us the Completeness strategy $3$. Translate $\mathbf{N}$ into $\mathcal{L}^U$ using the previously defined $h$, so that we obtain the translated versions of the normal modal logic axioms and inference rules (skipping \emph{Modus Ponens}, the classical propositional logic axioms, and Uniform Substitution):

\begin{itemize}
\item[$\mathrm{K}_h$] $\big{(} (\varphi \rightarrow \psi) \wedge \neg U (\varphi \rightarrow \psi) \big{)} \to \big{(} (\varphi \wedge \neg U \varphi) \rightarrow (\psi \wedge \neg U \psi) \big{)}$
\item[$\mathrm{N}_h$] $\varphi \rightarrow (\varphi \wedge \neg U \varphi)$
\item[$\mathrm{D}_h$] $\bot \vee \neg U \top$
\item[$Nec_h$] from $\vdash \varphi$ infer $\vdash \varphi \wedge \neg U \varphi$
\end{itemize}

Let $\mathbf{L} \supseteq \mathbf{N}$. Then, once again, by the structure of the translated axioms and inference rules, $h$ is theorem-preserving from $\mathbf{L}$ to $\mathbf{L}^U_h$.

\begin{lemma}\label{lemma h is surjective 2}
Let $\mathbf{L} \supseteq \mathbf{N}$. Then $h$ is surjective up to equivalence w.r.t. $\mathbf{L}$.
\end{lemma}

\begin{proof}
By induction on the complexity of $\varphi \in Form_{\mathcal{L}^U}$, we show there is $\psi \in Form_{\mathcal{L}^\Box}$ such that $\vdash^{\mathbf{L}_h^\lhd} \varphi \leftrightarrow h(\psi)$. By the compositionality of $h$, we may skip the non-modal cases, so let $\varphi = U \gamma$. By induction hypothesis, there is $\alpha \in Form_{\mathcal{L}^\Box}$ such that $\vdash^{\mathbf{L}^U_h} \gamma \leftrightarrow h(\alpha)$. Since $\mathbf{L} \supseteq \mathbf{N}$, we have $\mathbf{L}^U_h \supseteq \mathbf{N}_h$, so by $Nec_h$, $\vdash^{\mathbf{L}^U_h} (\gamma \leftrightarrow h(\alpha)) \wedge \neg U (\gamma \leftrightarrow h(\alpha))$. As before, by \emph{Modus Ponens} and $\mathrm{K}_h$, $\vdash^{\mathbf{L}_h^U} \neg U \gamma \leftrightarrow \neg U h(\alpha)$, which means $\vdash^{\mathbf{L}_h^U} U \gamma \leftrightarrow U h(\alpha)$.
\end{proof}

\begin{corollary}
Let $\mathbf{L} \supseteq \mathbf{N}$, $\mathbf{L}$ be canonical w.r.t. $\mathbb{C}_\mathbf{L}$, and $\mathbf{L}^U_h$ be sound w.r.t. $\mathbb{C}_\mathbf{L}$. Then, $\mathbf{L}^U_h$ is complete w.r.t. $\mathbb{C}_\mathbf{L}$.
\end{corollary}

\begin{proof}
By Proposition \ref{prop translation 2}, $h$ preserves truth in $\mathbb{C}_\mathbf{N}$. As we have seen, it is also theorem preserving, and by Lemma \ref{lemma h is surjective 2}, $h$ is surjective up to equivalence w.r.t. $\mathbf{L}$. Since $\mathbf{L}$ is canonical, and by Corollary \ref{corollary soudness 55}, $\mathbf{L}^U_h$ is sound w.r.t. $\mathbb{C}_\mathbf{L}$, we get, by Theorem \ref{theorem completeness 3}, that $\mathbf{L}_h^U$ is complete w.r.t. $\mathbb{C}_\mathbf{L}$.
\end{proof}

\section{Conclusion}\label{conclusion}\label{sec conclusion}

In this paper, we have unified several existing notions of definability and showed their equivalence. This clarification allows one to speak interchangeably about expressivity and definability of modal operators, and guarantees the existence of truth-preserving translations in both directions under the appropriate conditions. We have also introduced semantic insensitivity as the basis for the transfer of soundness transfer from more expressive languages to less expressive ones, and presented three -- and one more slight modification -- completeness transfer strategies for modal logics with a single operator under relational semantics. In two -- or three -- of the strategies the step of finding, by semantic insight, an accessibility condition tailored to the target logic can be avoided entirely: completeness may be obtained either directly by an appropriate translation, or by restricting the canonical model of $\mathbf{K}$ -- or of a chosen extension of $\mathbf{K}$ --, with no accessibility condition for the logic in question at all stated at all. The first strategy still requires an explicit accessibility condition, but one derived from the translation, and it yields completeness via a bounded morphism onto the canonical frame of the normal logic rather than a standalone canonical model argument.

Despite these advances, the connection between insensitivity and completeness remains open. None of our completeness strategies exploits insensitivity directly, and all apply only to operator languages with a single modal operator under relational semantics. Future work could aim to develop a genuine completeness method that employs insensitivity itself, and to extend the present strategies to languages with multiple modalities, to neighbourhood semantics, and to first-order modal languages.

We believe the study of semantic insensitivity offers a fruitful perspective on the expressivity of modal languages and a powerful tool for metatheoretic proofs. The strategies outlined here already simplify existing soundness and completeness arguments for logics of non-contingency, essence, and the new operators of Section 5, and we hope they will find wider application in the analysis of non-normal modal logics.

\bibliographystyle{plain}
\bibliography{bib_insensitivity}

\end{document}